\documentclass[11pt,a4paper]{article}
\usepackage{setspace}

\usepackage[margin=1.1in]{geometry}

\usepackage[T1]{fontenc}
\usepackage{amsmath}
\usepackage{amssymb}
\usepackage{bm,lscape}
\usepackage{longtable}
\usepackage{mathrsfs,dsfont}
\usepackage{graphicx}
\usepackage{eurosym}
\usepackage{tablefootnote}
\usepackage{footmisc}
\usepackage{subfig}
\usepackage{listings}
\usepackage{adjustbox}
\usepackage{booktabs}
 \usepackage[table,xcdraw]{xcolor}
\RequirePackage{natbib}
\usepackage[english]{babel}
\usepackage{amsmath,amsthm}
\usepackage{amssymb}
\usepackage{multirow}
\usepackage{bbm}
\usepackage{bm}
\usepackage{hyperref}
\usepackage{mathrsfs}
\usepackage{algorithm}
\usepackage{dsfont}
\usepackage{color}

\usepackage[noend]{algpseudocode}

\makeatletter
\def\BState{\State\hskip-\ALG@thistlm}
\makeatother

\theoremstyle{definition}

\newtheorem{corollary}{Corollary}
\newtheorem{theorem}{Theorem}

\def\bSig\mathbf{\Sigma}

\begin{document}
\doublespacing

\begin{center}
\textmd{\LARGE{\bfseries{{An econometric analysis of the Italian cultural supply }}}}
\end{center}
\medskip

\begin{center}
\large{Consuelo R. Nava\footnote{Department of Economics and Political Sciences, University of Aosta Valley, Italy and Department of Economic Policy, Catholic University of the Sacred Heart, Milan, Italy, \href{mailto:c.nava@univda.it}{c.nava@univda.it} },  
Maria Grazia Zoia\footnote{Department of Economic Policy, Catholic University of the Sacred Heart, Milan, Italy, \href{mailto:maria.zoia@unicatt.it}{maria.zoia@unicatt.it}}}
\end{center}


\smallskip

\begin{quote}
\begin{center}
\noindent {\sl Abstract} \end{center}
Price indexes in time and space is a most relevant topic in statistical analysis from both the methodological and the application side. In this paper a price index providing a novel and effective solution to price indexes over several periods and among several countries, that is in both a multi-period and a multilateral framework,  is devised. The reference basket of the devised index is the union of the intersections of the baskets of all periods/countries in pairs. As such, it provides a broader coverage than usual indexes. Index closed-form expressions and  updating formulas are provided and properties investigated. Last, applications with real and simulated data provide evidence of the performance of the index at stake.

\vspace{10pt}
\vspace{5pt}

\par
\end{quote}


\newpage

\section{Introduction}{\label{sec: intro}}
Multi-period and multilateral price indexes, used to compare sets of commodities over time and across countries, respectively, are of prominent interest for statisticians  \citep[see, e.g.,][]{Biggeri2010}. Several approaches to the problem have been carried out in the literature. 

One of them is the axiomatic approach \citep[see, e.g.,][and the references quoted therein]{balk1995axiomatic}, which rests on the availability of both quantities and prices, dealt with as independent random variables, and aims at obtaining price indexes, to enjoy suitable properties \citep[][]{fisher1921best, fisher1922making}.

A second approach hinges on the economic theory\footnote{This approach is also known as preference field approach or functional approach \citep{divisia1926indice}.} \citep[see, among others,][for a review]{diewert1979economic, caves1982economic} and rests on the idea that consumption choices come from the optimization of a utility function under budget constraints. Here, prices play the role of independent variables, while  quantities arise as solutions to an optimization problem in accordance with the decision maker's preference scheme.

A third approach is the stochastic one \citep[see][for a  review] {clements2006stochastic, diewert2010stochastic}, which can be traced back to the works of \cite{jevons1863serious, jevons1869depreciation} and \cite{edgeworth1887report, edgeworth1925memorandum}. Thanks to \citet{balk1980method, clements1987measurement}, this approach has recently been reappraised, and its role in inflation measurements duly acknowledged \citep[see, e.g.,][and references quoted therein]{zahid2010measuring}. In this framework, prices are assumed to be affected by  measurement errors whose bias effect must be duly minimized. 
The stochastic approach (hereafter, SA) turns out to be somewhat different from other approaches, insofar as it is closely related to regression theory \citep{theil1960best, clements1987measurement}. In fact, the SA enables the construction of tests and confidence intervals for price indexes, which provide useful pieces of information \citep{clements2006stochastic}. Furthermore, the SA has less limits than other approaches\footnote{The SA, differently from the index number theory originating from \cite{theil1967economics} does not need to account for the economic importance of single prices.} and clears the way to further extensions, as shown in \cite{diewert2004stochastic, diewert2005weighted, silver2009hedonic, rao2004country}. 

In this paper, we devise a multi-period/multilateral price index, MPL index henceforth, within the stochastic framework. The derivation of the MPL index, which is the solution to an optimization problem, calls for quantities and values of the commodities (not prices), like \citet{walsh1901measurement}. The reference basket, namely the set of commodities for all periods/countries, is made up of the union of the intersections of all the couples of year/country baskets in pairs. Namely, the price index of a commodity can be always computed once the latter is present in at least two periods/countries. Thus, the reference basket turns  out to be more representative than the ones commonly used by the majority of statistical agencies, which either align the reference basket to that of the first period, or make it tally with the intersection of the commodity sets of all periods/countries. Eventually, such a reference basket is likely to be scarcely representative of the commodities present in each period/country.	 
In this sense, just like hedonic \citep{pakes2003reconsideration}, GESKS \citep{balk2012price} and country/time-product-dummy (CPD/TPD) approaches with incomplete price tableaus \citep{rao2004country}, the MPL index does not drop any observation on the account of having no counterpart in the reference basket. The lack of a commodity in a period/country $t$ requires putting both its quantity and value equal to zero in that period/country. However, unlike the aforesaid approaches, the MPL index is built on quantities and values, not on prices. Neither any preliminary computation of binary price indexes, as in the GESKS approach \citep{ivancic2011scanner}, nor the use of any type of weighting matrix for dealing with missing values or quantities, as in the case of CPD/TPD indexes, are needed.

The updating of the MPL index is easy to accomplish and suitable formulas, tailored to the multi-period or multilateral nature of the data, are provided. 
In fact, while the inclusion of fresh values and quantities, of a set of commodities corresponding to an extra period, does not affect the previous values of the MPL index, the inclusion of a new country affects all former MPL indexes. Hence, two updating formulas have been proposed for the MPL index: one for the multi-period case and another for the multilateral case. A comparison of the said index to CPD/TPD index -- a multilateral/multi-period index that, like the MPL one, can be read as a solution to an optimization problem -- provides evidence of an easier implementation and greater accuracy of the former index.

To sum up, a threefold  novelty characterizes the paper. First, a price index, which proves effective either for the multi-period or the multilateral case, is devised.  
Second, updating formulas tailored to the multilateral and the multi-period version of the index are provided.
Third, the grater simplicity of use and accuracy of the said index is highlighted in comparison with well-known standard multilateral/multi-period indexes. 

The paper is organized as follows. In Section~\ref{sec: meth}, we briefly go over the SA to price indexes and point out how several indexes are solutions to optimization problems. In Section~\ref{sec:mpl}, within the SA, we devise the MPL index according to a minimum-norm criterion as well as its updating formulas for the multi-period and the multilateral cases, respectively. 
Section~\ref{subsec:prop} is devoted to the properties of the MPL index. Section~\ref{sec:empirical} provides an application of the MPL index to the Italian cultural supply data to shed light on its potential as both a multi-period and a multilateral index. To gain a better insight into the performance of the MPL index, a comparison is made of the CPD/TPD indexes by using both real and simulated data. Section~\ref{sec:conclusion} completes the paper with some concluding remarks and hints. For the sake of easier readability, an Appendix has been added with proofs and technicalities.  

\section{The stochastic approach revisited}{\label{sec: meth}}
In this section, we review the main features of the SA and show how several well-known price indexes can be obtained within this framework. 

The SA works out price indexes as solutions to an optimization problem consisting in finding a line (more generally a plane or a hyperplane) which lies as closely as possible to the points whose coordinates are the $N$ commodity prices in the $T$ periods under examination. Following \citet{theil1967economics} and \citet[][Ch. 3]{selvanathan1994index} and by assuming for exposition purposes $T=2$, the idea underlying the SA is that, in both the periods taken into account, all prices move almost proportionally. Namely,
\begin{equation}\label{eq:0}
p_{i2}\,\approx \lambda \,p_{i1}\,\,\,\,\,\,\forall\,i=1,2,\dots,N
\end{equation}
where  $p_{it}$  is the price of commodity $i$ in  period $t$ ($t=\{1,2\}$) and $\lambda$ is a scalar factor acting as price index. Eq.~\eqref{eq:0} can be  conveniently rewritten as follows 
\begin{equation}\label{eq:00}
\frac{p_{i2}}{p_{i1}}= \lambda +\eta_i\,\,\,\,\,\,\,\forall\,i=1,2,\dots,N
\end{equation} 
where $\eta_i$ are error terms which, as a rule, are assumed to be non-systematic and uncorrelated between commodities with variances, that can be either constant (with respect to commodities), that is
\begin{equation}
\label{eq:14}
var(\eta_i)=\sigma^2, \,\,\,\,\,\,\,\forall\,i=1,2,\dots,N,
\end{equation}
or not. In the latter case, variances are frequently specified as inversely proportional to the commodity budget share, namely 
\begin{equation}
\label{eq:15}
var(\eta_i)=\sigma^2\frac{1}{w_{i1}},\,\,\,\,\,\,\,\forall\,i=1,2,\dots,N.
\end{equation}
Here $w_{i1}=\frac{p_{i1}q_{i1}}{\boldsymbol{p}_1'\boldsymbol{q}_1}$ and $\boldsymbol{p}_1'\boldsymbol{q}_1$ is the total expenditure in the base period ($t=1$), being $\boldsymbol{p}_t$ and $\boldsymbol{q}_t$ the vectors of prices and quantities at time $t$. 
Eq.~\eqref{eq:00} can be written in compact notation as follows
\begin{equation}
\boldsymbol{p}_2*\boldsymbol{p}_1^{-H}=\lambda \boldsymbol{u}+\boldsymbol{\eta}
\label{eq:13}
\end{equation}
where $\boldsymbol{u}$ is an $N$-dimensional unit vector, $\boldsymbol{p}_1^{-H}$ is defined as the vector of the reciprocals of the non-null entries of  $\boldsymbol{p}_1$ and zeros otherwise, that is  
\begin{equation}
\label{eq:666}
\underset{(N,1)}{\boldsymbol{p}_1^{-H}}=\left[p_{i1}\right]^{-H}=\begin{cases} p_{i1}^{-1}\,\,\,\,\,\,\mbox{if}\,\,\,p_{i1}\neq 0
\\
0 \,\,\,\,\,\,\,\,\,\,\,\,\mbox{otherwise}
\end{cases}
\end{equation}
and  $\boldsymbol{\eta}$ is a vector of $N$ non-systematic and spherical random variables, unless otherwise specified.
Taking the Hadamard product of both sides of Eq.~\eqref{eq:13} by $\boldsymbol{p}_1$ , yields the following  
\begin{equation}
\label{eq:2}
 \boldsymbol{p}_2*\boldsymbol{p}_1^{-H}*\boldsymbol{p}_1=\lambda \boldsymbol{u}*\boldsymbol{p}_1+\boldsymbol{\eta}*\boldsymbol{p}_1=\lambda \boldsymbol{p}_1+\boldsymbol{\varsigma}.\footnote{The error terms $\boldsymbol\varsigma$ of Eq.~\eqref{eq:2}  are no longer homoschedastic. In particular, if the assumption in Eq.~\eqref{eq:14} holds, then the dispersion matrix of $\boldsymbol \varsigma$ takes the form                                                                          
\[\mathds{E}(\boldsymbol{\varsigma}\boldsymbol{\varsigma}')=\sigma^2\boldsymbol{D}\,\,\,\mbox{where}\,\,\,\boldsymbol{D}=[(\boldsymbol{p}_1*\boldsymbol{p}_1)\boldsymbol{u}']*\boldsymbol{I}_N\]
where $\boldsymbol{u}$ is the unit vector and $\boldsymbol{I}_N$ denotes the $N$-dimensional identity matrix. 
If Eq.~\eqref{eq:15} holds, then 
\[\mathds{E}(\boldsymbol{\varsigma}\boldsymbol{\varsigma}')=\sigma^2\boldsymbol{p}_1'\boldsymbol{q}_1\tilde{\boldsymbol{D}}\,\,\,\mbox{where}\,\,\,\tilde{\boldsymbol{D}}=[(\boldsymbol{p}_1*\boldsymbol{q}_1)^{-H}\boldsymbol{u}']*\boldsymbol{I}_N.\]}
\end{equation}
The appeal of the SA lays in the possibility of evaluating, besides point estimates,  also price index standard errors which increase as the relative price variability increases. The computation of price index standard errors allow to verify the intuitive notion that the less  prices move proportionally, the less precise are price index estimators. Further, standard errors prove useful to build confidence intervals for price indexes. 

The following theorem shows that several well-known price indexes can be seen as offspring of Eq.~\eqref{eq:2}. 
\begin{theorem}\label{th:1}
\textit{
Let $\boldsymbol{p}_t$ be a vector of the prices of $N$ commodities at time $t$ and $\boldsymbol{q}_t$ be the vector of the corresponding quantities. The Laspeyeres, Paasche, Marschall-Edgeworth and Walsh indexes are solutions to an optimization problem of the form
\begin{equation}
\min_{\lambda}\,\,\, || \boldsymbol{p}_2\,- \lambda  \boldsymbol{p}_1||_N=\min_{\lambda}\,\,\, ||\boldsymbol{e}||_N
\end{equation}
 where $||\,\cdot\,||_{N}=(\boldsymbol{e}'\boldsymbol{A}\boldsymbol{e})^{\frac{1}{2}}$  stands for a (semi)norm of the reference price index and $\boldsymbol{A}$ is a properly chosen non-negative definite matrix}.
\end{theorem}
\begin{proof}
See Appendix~\ref{app:proofs1}.
\end{proof}

\section{The MPL index as solution to an optimization problem }\label{sec:mpl}
In this section, a multi-period/multilateral price index is derived in the wake of the SA introduced in the previous section. The construction of this kind of index, MPL index hereafter, poses several issues, like the choice of the reference basket and its updating.
When the prices of commodity sets in two periods/countries are compared, the reference basket, $K_{\tau}$, is generally set to be a subset of the commodities of the first period/country ($K_{\tau}=K_{1}$), which is also assumed to be representative for the other period/country. 

The simplest solution is to take the ``intersection'' of the two baskets as reference basket, ($K_{\tau}= K_{1}\cap K_{2}$). When the comparison is among more than two periods/countries, statistical agencies generally align the reference basket to that of the first period/country ($K_{\tau}=K_{1}$), which turns out to play the role of base period/country.
Of course, this choice is somewhat arbitrary as it leaves open the basket updating problem  and it does not take into account links among baskets corresponding to couples, triples, $\dots$, of periods. An alternative approach sets the reference basket as the intersection of all the commodities considered in each single period ($K_{\tau}= \bigcap_{t} K_{t}$). As a result, the reference basket is likely to be partially representative of those commodities which are peculiar to each period. 

Taking the SA as the reference frame, we derive a multi-period/multilateral price index, satisfying a minimum-norm criterion, whose reference basket -- over a set of periods or across a set of countries -- is  the union of the intersections of the commodity baskets of various periods/countries, taken in pairs.\footnote{Some similarities arise with the chaining rule \citep{forsyth1981theory, von2001chain}. This is a specific type of temporal aggregation method based on the use of the complete time series when computing a price index from $0$ to $t$. The resulting price index is a measure of the cumulated effect of adjacent periods from 0 to 1, 1 to 2, $\dots$, $t-1$ to $t$. However, chain indexes leave unresolved the reference basket updating and are not applicable in a multilateral perspective, differently from the MPL index. Chain indexes compare the current and the previous periods in order to evaluate the evolution over many periods. For a comparison of this approach with the fixed base one see \cite{diewert2001consumer}} 
Such a reference basket proves to be an effective solution for several reasons. First, it is broader and more representative than the ones built on the intersection of  commodities which are present in all periods/countries. Second, the price index of a commodity is well defined, provided the latter is present in at least two baskets. The computation of the said price index requires quantities and values of the commodities, not their prices. The lack of a commodity in a given period/country $t$  simply entails setting both its quantity and value equal to zero in that period/country.
Figure~\ref{fig:1} shows the reference basket corresponding to the aforementioned approaches for the case of two and three periods/countries.
%
%
\begin{figure}[h!]
\begin{center}
\begin{doublespace}
\includegraphics[scale=0.27]{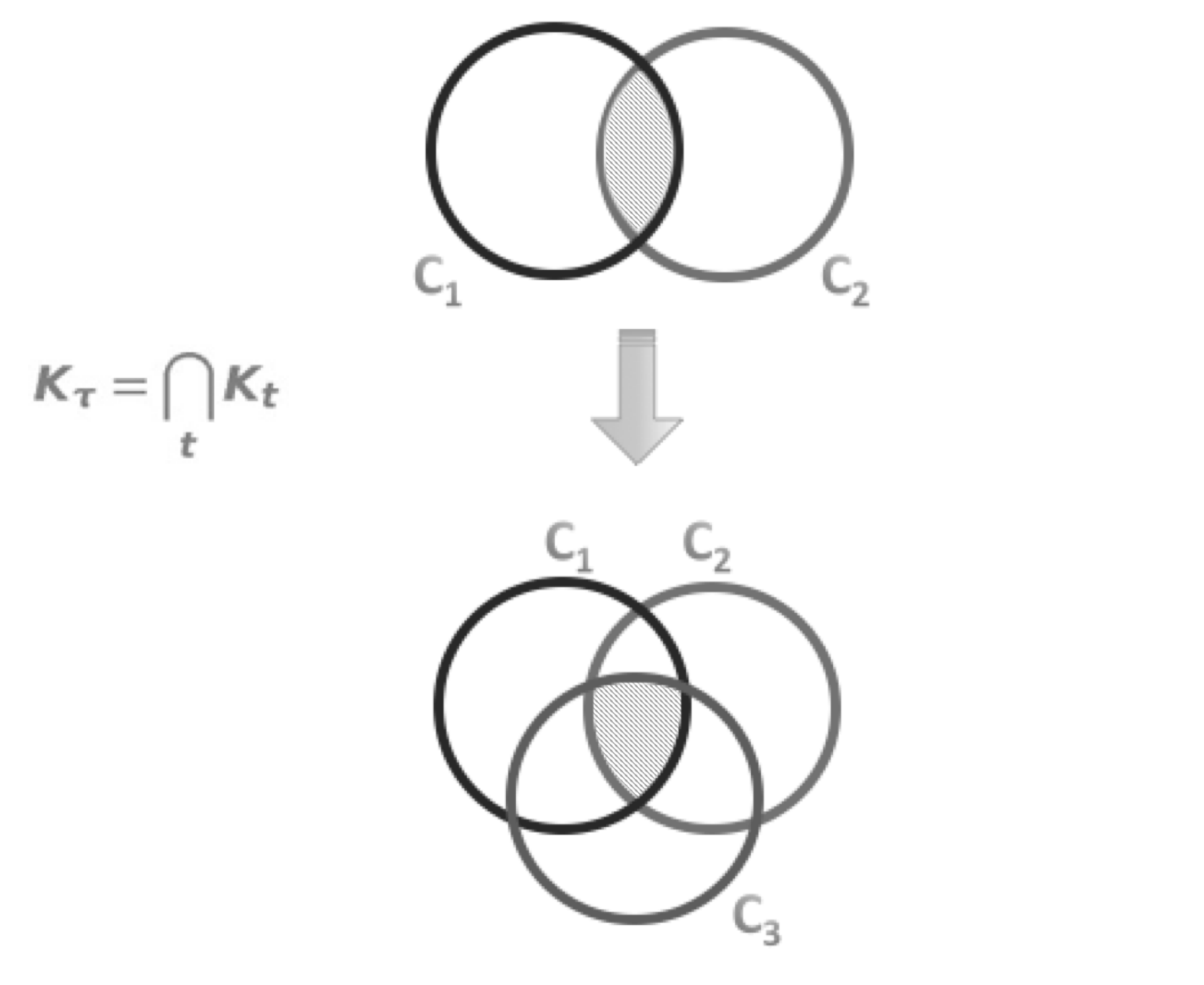} 
\includegraphics[scale=0.27]{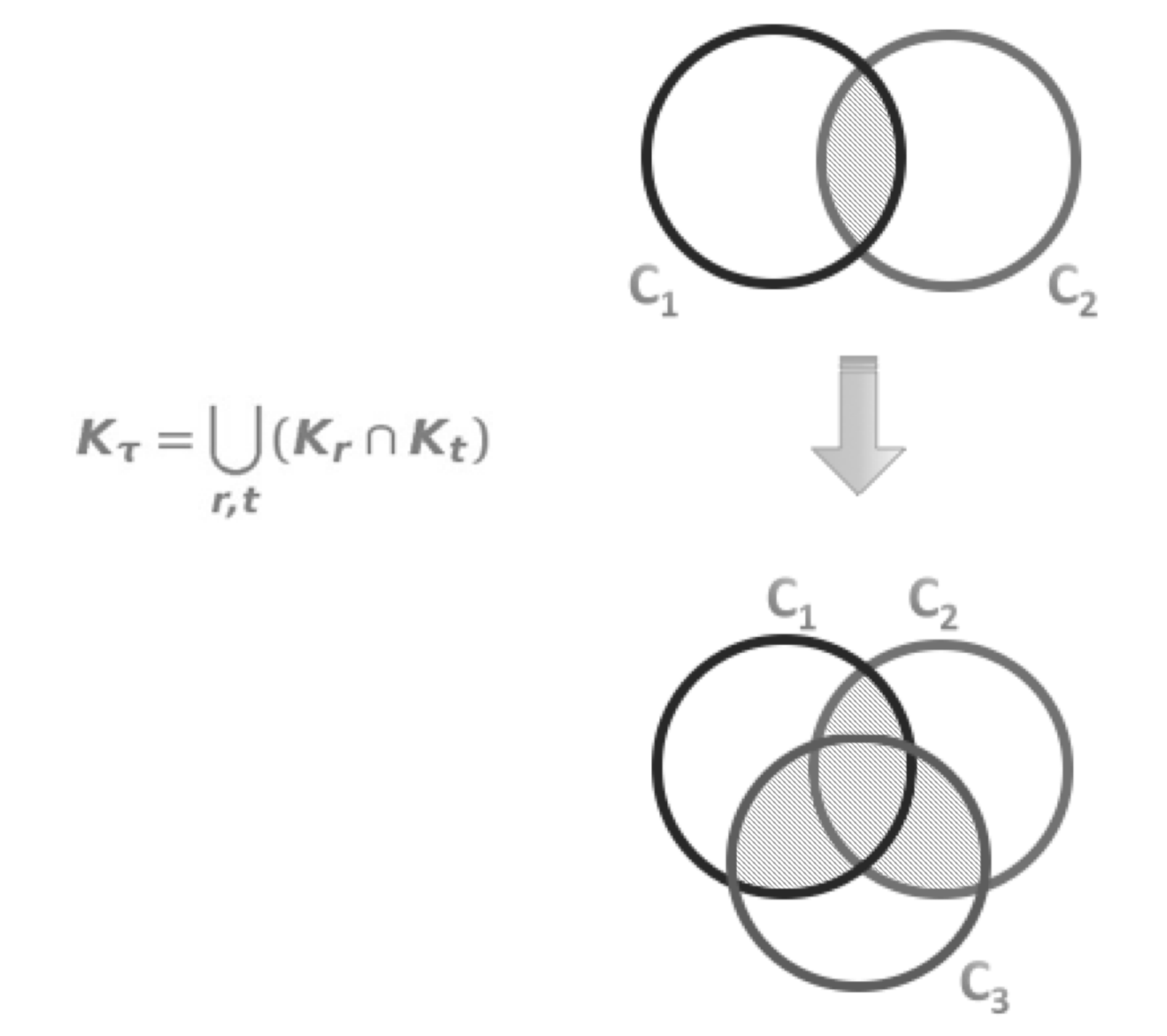}
\caption{
 The left-hand side panel shows the reference basket corresponding to the ``traditional'' approach. The right-hand side panel shows the reference basket corresponding to the MPL index. }
\label{fig:1}
\end{doublespace}
\end{center}
\end{figure} 
%
\noindent The MPL index hinges, first of all, on the idea that in each time/country $t$, the $N$ commodity prices  move proportionally to a  set of $N$ reference  prices, that is
\begin{equation}\label{eq:multi1}
\underset{(N,1)}{\boldsymbol{p}_t}\,\approx\, \lambda_t  \underset{(N,1)}{\boldsymbol{\tilde{p}}}\,\,\,\,\,\,\,\,\,\,\,\,\forall\,t=1,2,\dots,T
\end{equation}
or
\begin{equation}\label{eq:multi2}
\underset{(N,1)}{\boldsymbol{p}_t}\,=\, \lambda_t  \underset{(N,1)}{\boldsymbol{\tilde{p}}}+\underset{(N,1)}{\boldsymbol{\varsigma}_t}\,\,\,\,\,\,\forall\,t=1,2,\dots,T.
\end{equation}
Here $\boldsymbol{p}_t$ is the actual price vector of  the $N$ commodities at period/country $t$, $\tilde{\boldsymbol{p}}$ is the vector of (unknown (time invariant) reference prices, $\lambda_t$  is a scalar factor acting as price index at period/country $t$ and $\boldsymbol{\varsigma}_t$ is a vector of error terms. Eq.~\eqref{eq:multi2} can be viewed as a generalization of Eq.~\eqref{eq:2} for two periods/countries.
As per Eq.~\eqref{eq:multi2}, in each period/country $t$, the $N$ prices ($\boldsymbol{p}_t$)  can be represented by a point in a $N$-dimensional space. Accordingly, $N$ prices in $T$ periods/countries, $\boldsymbol{P}=[\boldsymbol{p}_{1},...,\boldsymbol{p}_{t},..,\boldsymbol{p}_{T}]$, can be represented by $T$ points in a $N$-dimensional space. If all prices move proportionally, these points would lie on a hyperplane, $\tilde{\boldsymbol{p}}$, and, in particular, on a straight line crossing the origin for
$T=2$. 

In general, this is only approximately true and the price ``line'' crossing the origin enjoys the property of fitting the observed price points, according to a  criterion which minimizes the deviations of the data from the ``line''. 
In compact notation, Eq.~\eqref{eq:multi1} can be more conveniently reformulated as follows
\begin{equation}\label{eq:multi3}
\underset{(N,T)}{\boldsymbol{P}_{}}\approx \,\underset{(N,T)}{\boldsymbol{\Pi}_{}}=\underset{(N,1)}{\tilde{\boldsymbol{p}}}\,\underset{(1,T)}{\boldsymbol{\lambda}_{}^{'}}
\end{equation}
where $\boldsymbol{\lambda}$ is the vector of the $T$ price indexes and  $\tilde{\boldsymbol{p}}$ is the vector of the $N$ (unknown) reference prices. According to Eq.~\eqref{eq:multi3}, the problem of determining a set of $T$ price indexes can be read as the problem of approximating the price matrix,  $\boldsymbol{P}$, with a matrix of unit rank, $\boldsymbol{\Pi}$, defined as the product of a vector of price indexes by a (unknown) reference price vector. Moving from prices to values, in light of Eq.~\eqref{eq:multi3}, the matrix $\boldsymbol{V}$ of the values of  $N$ commodities in $T$ periods/countries has the representation
\begin{equation}\label{eq:multi33}
\underset{(N,T)}{\boldsymbol{V}_{}}=\boldsymbol{P}\ast\boldsymbol{Q}\approx\boldsymbol{\Pi}\ast\boldsymbol{Q} \approx (\underset{(N,1)}{\tilde{\boldsymbol{p}}}\underset{(1,T)}{\boldsymbol{\lambda}'_{}})\ast\underset{(N,T)}{\boldsymbol{Q}_{}} 
\end{equation}
where $\boldsymbol{Q}$ is the matrix  of the quantities of $N$ commodities in $T$ periods/countries. 
Eq.~\eqref{eq:multi33} can be rewritten as
\begin{equation}\label{eq:gh}
[\boldsymbol{v}_{1},...,\boldsymbol{v}_{t},...,\boldsymbol{v}_{T}]\approx[\lambda_{1}\,\tilde{\boldsymbol{p}}_1\ast \boldsymbol{q}_{1},..,\lambda_{t}\,\tilde{\boldsymbol{p}}_t\ast \boldsymbol{q}_{t},...,\lambda_{T}\,\tilde{\boldsymbol{p}}_T\ast \boldsymbol{q}_{T}]
\end{equation}
where $\boldsymbol{v}_{t}$ and $\boldsymbol{q}_{t}$ are the $t^{th}$ columns of $\boldsymbol{V}$ and $\boldsymbol{Q}$ respectively, and $\lambda_{t}$ is the $t^{th}$ element of $\lambda_{t}$.
Eq.~\eqref{eq:gh} entails the following 
\begin{equation}
{\boldsymbol{v}_t}=\lambda_t {\tilde{\boldsymbol{p}}}\ast{\boldsymbol{q}_t}+{\bm{\varepsilon}_t}\,\,\,\,\,\forall\, t=1,2,\dots,T
\end{equation}
where ${\boldsymbol{\varepsilon}_t}$ is a non-systematic stochastic term. 
The above formula, taking into account the identity
\[\underset{(N,1)}{\tilde{\boldsymbol{p}}}*\underset{(N,1)}{\boldsymbol{q}_t}=\underset{(N,N)}{\boldsymbol{D}_{\tilde{\boldsymbol{p}}}} \underset{(N,1)}{\boldsymbol{q}_t}, \]
can be re-written as
\begin{equation}\label{eq:mod1}
\delta_t \underset{(N,1)}{\boldsymbol{v}_t} =\underset{(N,N)}{\boldsymbol{D}_{\tilde{\boldsymbol{p}}}}\,\underset{(N,1)}{\boldsymbol{q}_{t_{}}}+\underset{(N,1)}{\boldsymbol{\varepsilon}_t}\,\,\,\,\,\forall\, t=1,2,\dots,T                                            \end{equation}
where $\delta_t=(\lambda_t)^{-1}$ takes the role of the deflator,  and  $\boldsymbol{D}_{\tilde{\boldsymbol{p}}}$ is a diagonal matrix with diagonal entries equal to the elements of  $\tilde{\boldsymbol{p}}$ defined as in Eq.~\eqref{eq:multi2}.\footnote{The matrix $\boldsymbol{D}_{\tilde{\boldsymbol{p}}}$ is defined as follows:
\[\boldsymbol{D}_{\tilde{\boldsymbol{p}}}=\begin{bmatrix} \tilde{p}_{1} & 0 & \dots &0 \\ 0 & \tilde{p}_{2} & \dots &0
\\ \vdots & \vdots & \ddots &\vdots\\ 0 & 0 & \dots &\tilde{p}_{N} \end{bmatrix}.
\]\label{ftn:matrix}} 
Eq.~\eqref{eq:mod1} expresses the value, $v_{it}$, of each commodity $i$ at time $t$ (discounted by a factor  $\delta_t$) as the product between the (time invariant) reference  price, $\tilde{p}_i$, and the corresponding quantity, $q_{it}$, plus an error term, $\varepsilon_{it}$. Note that $\lambda_t=\delta_t^{-1}$ plays the role of the multi-period/multilateral price index in Eq.~\eqref{eq:mod1}.
Over $T$ periods/countries, the model can be written as
\begin{equation}\label{eq:1}
\underset{(N, T)}{\boldsymbol{V}_{}}\, \, \underset{(T, T)}{\boldsymbol{D}_{\boldsymbol{\delta}}}=  \underset{(N, N)}{\boldsymbol{D}_{\tilde{\boldsymbol{p}}}}\, \, \underset{(N, T)}{\boldsymbol{Q}_{}} + \underset{(N, T)}{\boldsymbol{E}_{}}
\end{equation}
where $\boldsymbol{D}_{\boldsymbol{\delta}}$ is a $T\times T$ diagonal matrix with diagonal entries equal to the elements of $\boldsymbol{\delta}$ (see Footnote~\ref{ftn:matrix}). 
\noindent Without lack of generality, we assume that the first period is the base period (that is  $\delta_1$=$\lambda_1=1$), and write the first equation separately from the others $T-1$. Accordingly, the system in Eq.~\eqref{eq:1} can be written as
\begin{equation}\label{eq:1b}
\left[\underset{(N,1)}{\boldsymbol{v}_{1},}\,\, \underset{(N,T-1)}{\boldsymbol{V}_1}\right] \begin{bmatrix} 1 &  \underset{(1,T-1)}{\boldsymbol{0}'}\\ \underset{(T-1,1)}{\boldsymbol{0}} & \underset{(T-1,T-1)}{\tilde{\boldsymbol{D}}_\delta}
\end{bmatrix} 
= \underset{}{\boldsymbol{D}_{\tilde{\boldsymbol{p}}}} \left[\underset{(N,1)}{\boldsymbol{q}_{1},}\,\, \underset{(N,T-1)}{\boldsymbol{Q}_1}\right]+\left[\underset{(N,1)}{\boldsymbol{\varepsilon}_{1},}\,\, \underset{(N,T-1)}{\boldsymbol{E}_1}\right]
\end{equation}
or as
\begin{equation}
\label{eq:17}
\begin{cases}
& \boldsymbol{v}_1=\boldsymbol{D}_{\tilde{\boldsymbol{p}}}\boldsymbol{q}_1+\boldsymbol{\varepsilon}_1
\\& \boldsymbol{V}_1 \tilde{\boldsymbol{D}}_\delta=\boldsymbol{D}_{\tilde{\boldsymbol{p}}}\boldsymbol{Q}_1+\boldsymbol{E}_1
\end{cases}
\end{equation}
which, after some computations\footnote{Use has been made of the relationships
\[vec(\boldsymbol{ABC})=(\boldsymbol{C}'\otimes \boldsymbol{A})\,vec(\boldsymbol{A}) \,\,\,\,\,\,\mbox{and}\,\,\,     \,\,\,                                                        
\underset{(N^2,1)}{vec(\boldsymbol{D}_{\boldsymbol{a}})}=\boldsymbol{R}'_N\,\underset{(N,1)}{\boldsymbol{a}}\]                                              
where $\boldsymbol{D}_{\boldsymbol{a}}$ is a diagonal matrix whose diagonal entries are the elements of the vector $\boldsymbol{a}$ and $\boldsymbol{R}_N$ is the transition matrix from the Kronecker to the Hadamard product \citep{faliva1996hadamard}.
},  can  be also expressed as
\begin{equation}
\label{eq:sist1}
\begin{cases}
&  \boldsymbol{v}_1=(\boldsymbol{q}_1' \otimes \boldsymbol{I}_N) \boldsymbol{R}_N' \tilde{\boldsymbol{p}}+\boldsymbol{\varepsilon}_1
\\
& \underset{(N(T-1),1)}{\boldsymbol{0}}=(\boldsymbol{I}_{T-1}\otimes (\boldsymbol{-V}_1))\boldsymbol{R}_{T-1}'\boldsymbol{\delta}+(\boldsymbol{Q}_1'\otimes \boldsymbol{I}_N)\boldsymbol{R}_N'\tilde{\boldsymbol{p}}+\boldsymbol{\eta}\
\end{cases}.
\end{equation}
Here, $\boldsymbol{\eta}= vec(\boldsymbol{E}_1)$, $\boldsymbol{\delta}$ is a vector whose elements are the diagonal entries of $\tilde{\boldsymbol{D}}_{\delta}$ as specified in Eq.~\eqref{eq:1b} and $\boldsymbol{R}_j$ denotes the transition matrix from the Kronecker to the Hadamard product.\footnote{The matrix $\boldsymbol{R}_{j}'$ is defined as follows \[\underset{(j^2\times j)}{\boldsymbol{R}_j'} =
\begin{bmatrix} \underset{(1,j)}{\boldsymbol{e}_{1}}\otimes\underset{(1,j)}{\boldsymbol{e}_{1}} &
\underset{(j,1)}{\boldsymbol{e}_{2}}\otimes\underset{(j,1)}{\boldsymbol{e}_{2}}
& \dots  & \underset{j,1)}{\boldsymbol{e}_{j}}\otimes\underset{(1,j)}{\boldsymbol{e}_{j}} \end{bmatrix},
\]
where \noindent  $\boldsymbol{e}_i$  represents  the $N$ dimensional $i$-th elementary vector. } 
The vector $\boldsymbol{\delta}$, whose (non-null) entries are the reciprocals of the elements of the price index 
$\boldsymbol{\lambda}$\footnote{The reciprocal of the deflator vector,  $\boldsymbol{\delta}$, is defined in Eq.~\eqref{eq:666}.},  
can be obtained by estimating Eq.~\eqref{eq:sist1} with the ordinary least squares (OLS), under the assumption that 
\begin{equation}
\mathds{E}(\boldsymbol{\mu}\boldsymbol{\mu}')=\sigma^2 \boldsymbol{I}_N \otimes \boldsymbol{I}_T=\sigma^2 \boldsymbol{I}_{NT} 
\,\,\,\mbox{where}\,\,\,\boldsymbol{\mu}=\begin{bmatrix}\varepsilon_1 & \boldsymbol{\eta}\end{bmatrix}'.\end{equation}
In this connection, we state the following result.
\begin{theorem}\label{th:2}
\textit{The MPL index can be obtained as reciprocal of the OLS estimate of the deflator vector,
$\boldsymbol{\delta}$,  in the system in Eq.~\eqref{eq:sist1}. This OLS estimate is given by
\begin{equation}\label{eq:delta}\begin{split}
\underset{(T-1,1)}{\hat{{\boldsymbol{\delta}}}}=&\left\{
(\boldsymbol{I}_{T-1}\ast\boldsymbol{V}_1'\boldsymbol{V}_1)-(\boldsymbol{Q}_1'\ast\boldsymbol{V}_1')\left[
(\boldsymbol{q}_1\boldsymbol{q}_1'+\boldsymbol{Q}_1\boldsymbol{Q}_1')\ast\boldsymbol{I}_N
\right]^{-1}(\boldsymbol{Q}_1\ast\boldsymbol{V}_1)
\right\}^{-1}\\&(\boldsymbol{Q}_1'\ast\boldsymbol{V}_1')\left[
(\boldsymbol{q}_1\boldsymbol{q}_1'+\boldsymbol{Q}_1\boldsymbol{Q}_1')\ast\boldsymbol{I}_N
\right]^{-1}(\boldsymbol{q}_1\ast \boldsymbol{v}_1).\end{split}\end{equation}
}
\end{theorem}
\begin{proof}
See Appendix~\ref{app:proofs2}.
\end{proof}
\noindent
The case where $T=2$ is worth considering because it sheds light on the index structure. In this case, the price index turns out to be a ratio of weighted price averages, with weights depending on the harmonic means of the squared quantities as stated in the following corollary.  
\begin{corollary}\label{cor:1}
\textit{When $T=2$, the MPL index, $\widehat{\lambda}$, becomes
\begin{equation}
\label{eq:ffrt0}
\hat{{\lambda}}=\frac{\sum_{i=1}^Np_{i2}\pi_i}{\sum_{i=1}^Np_{i1}\pi_i}
\end{equation}
where $\pi_i= 2 p_{i2}\frac{q_{i1}^2q_{i2}^2}{q_{i1}^2+q_{i2}^2}$ and $p_{it}=\frac{v_{it}}{q_{it}}$. The price index in Eq.~\eqref{eq:ffrt0} can be also   expressed as a convex linear  combination of prices, that is
\begin{equation}
\label{eq:ffrt}
\hat{{\lambda}}=\sum_{i=1}^N\frac{p_{i2}}{p_{i1}}\tilde{\pi}_i
\end{equation}
with weights $\tilde{\pi}_{i}$ given by 
\begin{equation}
\tilde{\pi}_{i}=\frac{v_{i1}v_{i2}\,\frac{q_{i1}q_{i2}}{q_{i1}^2+q_{i2}^2}}{\sum_{i=1}^N v_{i1}v_{i2}\,\frac{q_{i1}q_{i2}}{q_{i1}^2+q_{i2}^2}}.
\end{equation}
Furthermore, by setting $\tilde{\boldsymbol{q}}_t$=$\boldsymbol D^{-1/2}\boldsymbol q_t$ with $\boldsymbol D$ specified as follows
\begin{equation}
\boldsymbol D=\begin{bmatrix} q_{11}^2+q_{12}^2 & 0 & \dots &0 \\ 0 & q_{21}^2+q_{22}^2  & \dots &0
\\ \vdots & \vdots & \ddots &\vdots\\ 0 & 0 & \dots &q_{N1}^2+q_{N2}^2 \end{bmatrix},
\end{equation}
the $MPL$ index can  be also written in compact form as follows}
\begin{equation}
\hat{\lambda}=\frac{\left(\tilde{\boldsymbol q}_1\ast \boldsymbol v_2\right)'\left(\tilde{\boldsymbol q}_1\ast \boldsymbol v_2\right)}{\left(\tilde{\boldsymbol q}_2\ast \boldsymbol v_2\right)'\left(\tilde{\boldsymbol q}_1\ast \boldsymbol v_1\right)}.
\end{equation}
\end{corollary}
\begin{proof}
See Appendix~\ref{app:proofs3}.
\end{proof}
\noindent As a by-product of Theorem~\ref{th:2}, we state the following result.
\begin{corollary}\label{cor:3}
\textit{The variance-covariance matrix of the deflator vector, $\boldsymbol{\widehat{\delta}}$, given in Theorem~\ref{th:2} is 
\begin{equation}
Var(\boldsymbol{\hat{\delta}})=\sigma^2[\boldsymbol{I}_{T-1}*\boldsymbol{V}_1'\boldsymbol{V}_1]^{-1}.
\end{equation}
\noindent The $t^{th}$  diagonal entry of the above matrix provides the variance of the deflator in the $t^{th}$period/country, given by
\begin{equation}
var({\hat{\delta}_{t}})=\sigma^2[\boldsymbol{v}_t'\boldsymbol{v}_t]^{-1}
\end{equation}
where $\boldsymbol{v}_t$ denotes the $t^{th}$ column of the matrix $\boldsymbol{V}_1$.}
\end{corollary}
\begin{proof}
See Appendix~\ref{app:proofs6}.
\end{proof}
\noindent As for  the deflator vector $\boldsymbol{\hat{\delta}}$, its moments and confidence intervals can be easily obtained within the theory of linear regression models, from the result given in Corollary~\ref{cor:3}. As the price index vector $\hat{\boldsymbol{\lambda}}$  turns out to be the reciprocal of the said deflator (see Eq.~\ref{eq:666}), its  statistical behavior can be derived from the former, following the arguments put forward, for example, in \citet{geary1930frequency, curtiss1941distribution} and \citet{marsaglia1965ratios}, merely to quote a few, on ratios (in particular reciprocals) of random variables. 

The following corollary provides an approximation of the variance of $\hat{\lambda}_t$, obtained by using the first Taylor expansion of the variance of a ratio of two random variables. 

\begin{corollary}
\textit{The variance of the MPL index ${\hat{\lambda}_{t}}$ is
\begin{equation}
var({\hat{\lambda}_{t}}) \approx \frac{var({\hat{\delta}_{t}})}{E({\hat{\delta}_{t}}^4)}
=\frac{\widehat{\sigma}^2}{[\boldsymbol{v}_t'\boldsymbol{v}_t]E({\hat{\delta}_{t}}^4)} \,\,\,\,\,\,\forall\,t=1,\dots,T
\end{equation}
\noindent In the above equation $\widehat{\sigma}^2=\frac{\widehat{\boldsymbol\mu}'\widehat{\boldsymbol \mu}}{NT-(N+T-1)}$ where $\widehat{\boldsymbol\mu}$ are the OLS residuals of the equation system \eqref{eq:sist1}.}
\end{corollary}
\begin{proof}
See \citet[p. 351]{stuard1994kendall} and \citet[p. 69]{elandt1980survival}. 
\end{proof}

\noindent The following two theorems provide updating formulas for the  price index $\boldsymbol{\hat{\lambda}}$. The former proves suitable when the index is used as a multilateral price index, while the latter is appropriate when it is employed as a multi-period index. In the former case,  values and quantities of the commodities included in the reference basket are assumed available for an additional $T+1$ country. In the latter case, it is supposed that values and quantities of the commodities included in the reference basket become available at time $T+1$. 

\begin{theorem}\label{th:4}
\textit{
Should the values and quantities of $N$ commodities of a reference basket become available for a new additional country, say the $T+1$-th, then, the updated multilateral version of the MPL index, $\boldsymbol{\hat{\lambda}}$, turns out to be the vector of the reciprocals, as defined in Eq.~\eqref{eq:666}, of the following deflator vector
\begin{footnotesize}
\begin{equation}\label{eq:th4}
\begin{split}
\underset{(T,1)}{\hat{\boldsymbol{\delta}}}=&\\
=&\left\{
\begin{bmatrix} 
\boldsymbol{I}_{T-1}\ast\boldsymbol{V}'_1\boldsymbol{V}_1 & \boldsymbol{0}  \\ 
\boldsymbol{0} & \boldsymbol{v}'_{T+1}\boldsymbol{v}_{T+1}  \end{bmatrix}-
\begin{bmatrix} 
\boldsymbol{Q}'_1\ast\boldsymbol{V}'_1  \\ 
\boldsymbol{v}'_{T+1}\ast\boldsymbol{q}'_{T+1}  \end{bmatrix}
\begin{bmatrix} 
\left(\boldsymbol{q}_1\boldsymbol{q}'_1+\boldsymbol{Q}_1\boldsymbol{Q}'_1 +\boldsymbol{q}_{T+1}\boldsymbol{q}'_{T+1}\right) \ast\boldsymbol{I}_N\end{bmatrix}^{-1} \right.
\\&
\left.\begin{bmatrix} 
\boldsymbol{Q}_1\ast\boldsymbol{V}_1  & 
\boldsymbol{v}_{T+1}\ast\boldsymbol{q}_{T+1}  \end{bmatrix}
\right\}^{-1}
\begin{bmatrix} 
\boldsymbol{Q}'_1\ast\boldsymbol{V}'_1  \\ 
\boldsymbol{v}'_{T+1}\ast\boldsymbol{q}'_{T+1}  \end{bmatrix}
\begin{bmatrix} 
\left(\boldsymbol{q}_1\boldsymbol{q}'_1+\boldsymbol{Q}_1\boldsymbol{Q}'_1 +\boldsymbol{q}_{T+1}\boldsymbol{q}'_{T+1}\right) \ast\boldsymbol{I}_N\end{bmatrix}^{-1}
\left(\boldsymbol{q}_1\ast\boldsymbol{v}_1\right).
\end{split}
\end{equation}
\end{footnotesize}
\noindent Here the symbols are defined as in Theorem~\ref{th:2} and $\boldsymbol{v}_{T+1}$, $\boldsymbol{q}_{T+1}$ denote the vector of values and quantities of  $N$ commodities of the new $T+1$-th country, respectively.
}
\end{theorem}
\begin{proof}
See Appendix~\ref{app:proofs7}.
\end{proof}

\begin{theorem}\label{th:5}
\textit{
Should the values and quantities of $N$ commodities of a reference basket become available for  time $T+1$,  then, the updated value $\hat{\lambda}_{T+1}$ of the multi-period version of the MPL index at time $T+1$ turns out to be the reciprocal of the  deflator value at time $T+1$
}
\begin{equation}\label{eq:update1}
\begin{split}
\underset{(1,1)}{\hat{\delta}_{T+1}}=&\left\{
\boldsymbol{v}_{T+1}'\boldsymbol{v}_{T+1}-\left(\boldsymbol{q}_{T+1}'\ast\boldsymbol{v}_{T+1}'\right)\left[
(\boldsymbol{Q}\boldsymbol{Q}'+\boldsymbol{q}_{T+1}\boldsymbol{q}_{T+1}')\ast\boldsymbol{I}_{N}
\right]^{-1}\left(\boldsymbol{q}_{T+1}\ast\boldsymbol{v}_{T+1}\right)
\right\}^{-1} \\&\left(\boldsymbol{q}_{T+1}'\ast\boldsymbol{v}_{T+1}'\right)\left[
(\boldsymbol{Q}\boldsymbol{Q}'+\boldsymbol{q}_{T+1}\boldsymbol{q}_{T+1}')\ast\boldsymbol{I}_{N}
\right]^{-1}\left(\boldsymbol{Q}\ast\boldsymbol{V}\right)\tilde{{\boldsymbol{\delta}}}
\end{split}
\end{equation}
\noindent \textit{where $\tilde{{\boldsymbol{\delta}}}'=[1, \hat{{\boldsymbol{\delta}}}]'$ and $\hat{{\boldsymbol{\delta}}}$ is defined as in Eq.~\eqref{eq:delta}.}
\end{theorem}
\begin{proof}
See Appendix~\ref{app:proofs8}.
\end{proof}
\noindent Figure~\ref{fig:2} highlights the difference between the updating process of the deflator, and thus of the price index, depending on whether it is used in the multilateral or in the multi-period case.    

\begin{figure}[h!]
\begin{center}
\includegraphics[scale=0.3]{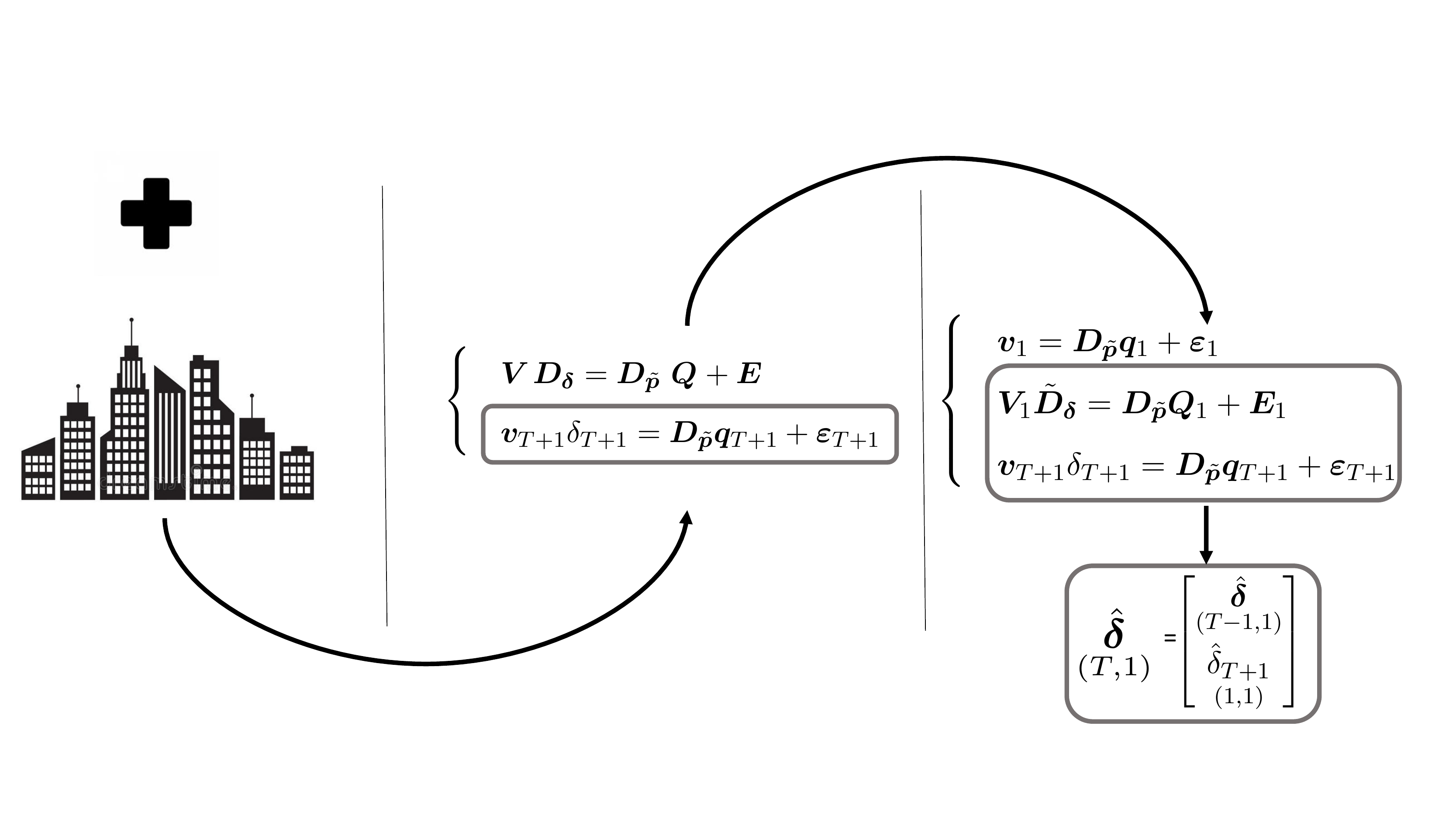}
\includegraphics[scale=0.3]{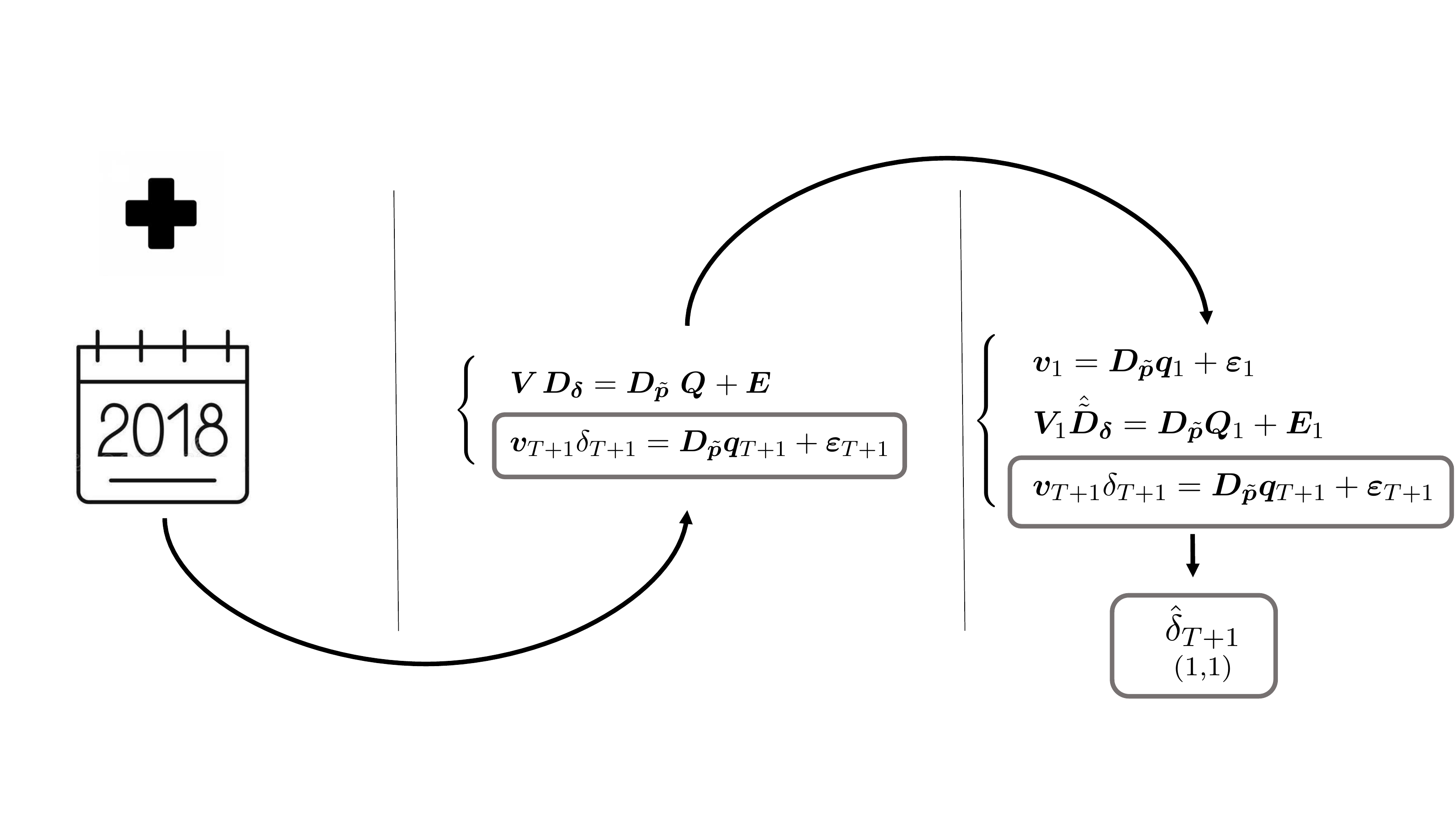} 
\caption{
The top panel shows the ratio of the updating formula for the multilateral version of the MPL index (see Theorem~\ref{th:4}); the bottom panel shows the ratio of the updating formula of the the multi-period version of the MPL index (see Theorem~\ref{th:5}).  
} 
\label{fig:2}
\end{center}
\end{figure}


\section{Properties of the MPL index}\label{subsec:prop}

Let us assume for simplicity $T=2$ and denote with $\hat{\lambda}(\boldsymbol{p}_{1},\boldsymbol{p}_{2},\boldsymbol{q}_{1},\boldsymbol{q}_{2})$ a generic index number where $\boldsymbol{p}_{t}$ and $\boldsymbol{q}_{t}$ are prices and quantities at time $t$. Without lack of generality, $t=1$ is assumed as the base period.
Following \citet{predetti2006numeri} and \citet{martini1992numeri}, the main properties of an index number can be summarized as follows:
\begin{enumerate}
\item \textit{Strong identity}: $\hat{\lambda}(\boldsymbol{p}_{2},\boldsymbol{p}_{2},\boldsymbol{q}_{1},\boldsymbol{q}_{2})=1$.
\item \textit{Commensurability}: $\hat{\lambda}(\boldsymbol{\gamma}\ast\boldsymbol{p}_{1},\boldsymbol{\gamma}\ast\boldsymbol{p}_{2},\boldsymbol{\gamma}^{-1}\ast\boldsymbol{q}_{1},\boldsymbol{\gamma}^{-1}\ast\boldsymbol{q}_{2})=\hat{\lambda}(\boldsymbol{p}_{1},\boldsymbol{p}_{2},\boldsymbol{q}_{1},\boldsymbol{q}_{2})$ with $\underset{(N,1)}{\boldsymbol{\gamma}}>\textbf{0}$ where $\boldsymbol{\gamma}^{-1}\ast \boldsymbol{q}_{t}=[q_{1t}/\gamma_1\,\,\dots\,\,q_{Nt}/\gamma_N]'$.
\item \textit{Proportionality}:  $\hat{\lambda}(\boldsymbol{p}_{1},\alpha\boldsymbol{p}_{2},\boldsymbol{q}_{1},\boldsymbol{q}_{2})=\alpha\,\hat{\lambda}(\boldsymbol{p}_{1},\boldsymbol{p}_{2},\boldsymbol{q}_{1},\boldsymbol{q}_{2})$ with $\alpha>0$.
\item \textit{Dimensionality}:   $\hat{\lambda}(\alpha\boldsymbol{p}_{1},\alpha\boldsymbol{p}_{2},\boldsymbol{q}_{1},\boldsymbol{q}_{2})=\hat{\lambda}(\boldsymbol{p}_{1},\boldsymbol{p}_{2},\boldsymbol{q}_{1},\boldsymbol{q}_{2})$ with $\alpha>0$.
\item \textit{Monotonicity}: $\hat{\lambda}(\boldsymbol{p}_{1},\boldsymbol{k}\ast\boldsymbol{p}_{2},\boldsymbol{q}_{1},\boldsymbol{q}_{2})>\hat{\lambda}(\boldsymbol{p}_{1},\boldsymbol{p}_{2},\boldsymbol{q}_{1},\boldsymbol{q}_{2})$ and $\hat{\lambda}(\boldsymbol{k}\ast\boldsymbol{p}_{1},\boldsymbol{p}_{2},\boldsymbol{q}_{1},\boldsymbol{q}_{2})<\hat{\lambda}(\boldsymbol{p}_{1},\boldsymbol{p}_{2},\boldsymbol{q}_{1},\boldsymbol{q}_{2})$ with $\underset{(N,1)}{\boldsymbol{k}}>\boldsymbol{u}$ where $\boldsymbol{u}$ is the unit vector.
\end{enumerate}
\noindent We now prove the above properties for the MPL index in Eq.~\eqref{eq:ffrt}, which will be denoted with the extended notation 
$\hat{{\lambda}}(\boldsymbol{p}_{1},\boldsymbol{p}_{2},\boldsymbol{q}_{1},\boldsymbol{q}_{2})$.
\begin{enumerate}
\item \textit{Strong identity}: 
taking the price vectors as equal in the two periods, it can be easily proved that
\begin{equation}
\hat{{\lambda}}(\boldsymbol{p}_{2},\boldsymbol{p}_{2},\boldsymbol{q}_{1},\boldsymbol{q}_{2})=\frac{\sum_{i=1}^Np_{i2}\pi_i}{\sum_{i=1}^Np_{i2}\pi_i}=1\,\,\,\mbox{where}\,\,\,\pi_{i}=2p_{i2}\frac{q_{i1}^2q_{i2}^2}{q_{i1}^2+q_{i2}^2}.
\end{equation}
\item \textit{Commensurability}: this property is trivially verified upon noting that, when each element $p_{it}$ of the price vectors are multiplied by a positive constant $\gamma_{i}$ and the corresponding quantity $q_{it}$ is divided by $\gamma_i$, then the weights $\tilde{\pi}_i$ become
$\tilde{\pi}_i=\frac{{\pi}_i}{\gamma_i}$.
Accordingly,
\begin{equation}
\hat{{\lambda}}(\boldsymbol{\gamma}\ast\boldsymbol{p}_{1},\boldsymbol{\gamma}\ast \boldsymbol{p}_{2},\boldsymbol{\gamma}^{-1}\ast\boldsymbol{q}_{1},\boldsymbol{\gamma}^{-1}\ast\boldsymbol{q}_{2})=\frac{\sum_{i=1}^N \gamma_i p_{i2}(\gamma_i ^{-1}\pi_i)}{\sum_{i=1}^N \gamma_i p_{i1} (\gamma_i^{-1}\pi_i)}=\frac{\sum_{i=1}^Np_{i2}\pi_i}{\sum_{i=1}^Np_{i1}\pi_i}.
\end{equation}
\item \textit{Proportionality}: this property is satisfied upon noting that $\breve{\pi}_i=\alpha{\pi}_i$.
Accordingly, \begin{equation}
\hat{{\lambda}}(\boldsymbol{p}_{1},\alpha\boldsymbol{p}_{2},\boldsymbol{q}_{1},\boldsymbol{q}_{2})=\frac{\sum_{i=1}^N \alpha p_{i2}\breve{\pi}_i}{\sum_{i=1}^N p_{i1} \breve{\pi}_i}=\alpha\frac{\sum_{i=1}^N p_{i2}\pi_i}{\sum_{i=1}^N p_{i1} \pi_i}.
\end{equation}  
\item \textit{Dimensionality}: the proof of this property follows the same argument used to prove homogeneity. Upon noting that  $\breve{\pi_i}=\alpha{\pi}_i$,  it follows that
\begin{equation}
\hat{{\lambda}}(\alpha\boldsymbol{p}_{1},\alpha\boldsymbol{p}_{2},\boldsymbol{q}_{1},\boldsymbol{q}_{2})=\frac{\sum_{i=1}^N\alpha p_{i2}\breve{\pi}_i}{\sum_{i=1}^N\alpha p_{i1}\breve{\pi}_i}=\frac{\sum_{i=1}^N p_{i2}\pi_i}{\sum_{i=1}^Np_{i1}\pi_i}.
\end{equation}
\item \textit{Monotonicity}: if we consider $\boldsymbol{k}>\boldsymbol{u}$, then the associated price index  is \[\hat{{\lambda}}(\boldsymbol{p}_{1},\boldsymbol{k}\ast\boldsymbol{p}_{2},\boldsymbol{q}_{1},\boldsymbol{q}_{2})=\frac{\sum_{i=1}^N k_i p_{i2}\pi_i}{\sum_{i=1}^N  p_{i1}\pi_i}>\frac{\sum_{i=1}^N p_{i2}\pi_i}{\sum_{i=1}^N  p_{i1}\pi_i}=\hat{{\lambda}}(\boldsymbol{p}_{1},\boldsymbol{p}_{2},\boldsymbol{q}_{1},\boldsymbol{q}_{2})\] if $k_i^2>k_j$ and $k_j^2>k_i$ for every $i\neq j$ given $i,j=1,\dots,N$.
While if we consider $\boldsymbol{k}\ast\boldsymbol{p}_{1}$, taking into account that the weight vector  $\boldsymbol{\pi}$ is independent from $\boldsymbol{k}$, the following holds \[\hat{{\lambda}}(\boldsymbol{k}\ast\boldsymbol{p}_{1},\boldsymbol{p}_{2},\boldsymbol{q}_{1},\boldsymbol{q}_{2})=\frac{\sum_{i=1}^N  p_{i2}\pi_i}{\sum_{i=1}^N  k_i p_{i1}\pi_i}<\frac{\sum_{i=1}^N p_{i2}\pi_i}{\sum_{i=1}^N  p_{i1}\pi_i}=\hat{{\lambda}}(\boldsymbol{p}_{1},\boldsymbol{p}_{2},\boldsymbol{q}_{1},\boldsymbol{q}_{2}).\] 
\end{enumerate}
It is worth noticing that the MPL index enjoys also the following properties:
\begin{itemize}
\item \textit{Positivity}: this property follows straightforward given that the index is the sum of ratios of non-negative quantities ($\pi_i, p_{i2}, p_{i1}\,\,\forall\,i=1,2,\dots, N$).
\item \textit{Inverse proportionality in the base period}: let consider a vector of prices in the base period  $\alpha \,\boldsymbol{p}_{1}$ with $\alpha>0$. Under this case, $\pi_i$ turns out to be independent from $\alpha$ and the associated index price proves to be proportional to $\hat{\lambda}(\boldsymbol{p}_{1},\boldsymbol{p}_{2},\boldsymbol{q}_{1},\boldsymbol{q}_{2})$
\[\hat{\lambda}(\alpha\boldsymbol{p}_{1},\boldsymbol{p}_{2},\boldsymbol{q}_{1},\boldsymbol{q}_{2})=\frac{\sum_{i=1}^N p_{i2} \pi_i}{\sum_{i=1}^N\alpha p_{i1} \pi_i}=\frac{1}{\alpha} \frac{\sum_{i=1}^N p_{i2} \pi_i}{\sum_{i=1}^Np_{i1} \pi_i}=\frac{1}{\alpha} \hat{\lambda}(\boldsymbol{p}_{1},\boldsymbol{p}_{2},\boldsymbol{q}_{1},\boldsymbol{q}_{2}).\]
\item \textit{Commodity reversal property}: it follows straightforward that the index price is invariant with respect to any permutation $(i)$:
\[\hat{\lambda}(\boldsymbol{p}_{1},\boldsymbol{p}_{2},\boldsymbol{q}_{1},\boldsymbol{q}_{2})= \frac{\sum_{(i)=1}^{(N)} p_{(i)2} \pi_{(i)}}{\sum_{(i)=1}^{(N)}p_{(i)1} \pi_{(i)}}= \frac{\sum_{i=1}^N p_{i2} \pi_{i}}{\sum_{i=1}^Np_{i1} \pi_{i}}.\]
\item \textit{Quantity reversal test}:  a change in the quantity order  only affects  $\pi_i$ that remains invariant $\forall\, i=1,\dots,N$. Therefore the index price $\hat{\lambda}$ does not change. 
\item \textit{Base reversibility (symmetric treatment of time)}: for one commodity, that is $N=1$, it is easy to prove that $\hat{\lambda}(p_{1},p_{2}, q_1,q_2)=\frac{p_2}{p_1}$ and, thus,  $\hat{\lambda}(p_{2},p_{1}, q_2,q_1)=\frac{p_1}{p_2}=\hat{\lambda}(p_{1},p_{2}, q_1,q_2)^{-1}$.
\item \textit{Transitivity}: for $N=1$, $\hat{\lambda}(p_{1},p_{2}, q_1,q_2)=\frac{p_2}{p_1}$, $\hat{\lambda}(p_{2},p_{3}, q_2,q_3)=\frac{p_3}{p_2}$, $\hat{\lambda}(p_{1},p_{3}, q_1,q_3)=\frac{p_3}{p_1}$ therefore, $\hat{\lambda}(p_{2},p_{3}, q_2,q_3)=\hat{\lambda}(p_{1},p_{2}, q_1,q_2) \,\cdot\,\hat{\lambda} (p_{2},p_{3}, q_2,q_3)$.
\item \textit{Monotonicity}: if $\boldsymbol{p}_{2}=\beta \boldsymbol{p}_{1}$ then
\begin{equation}
\tilde{\pi}_i=\beta\left(2p_{i1}\frac{q_{i1}^2q_{i2}^2}{q_{i1}^2+q_{i2}^2}\right)=\beta \gamma_i,
\end{equation}
 and the associated index price turns out to be equal to $\beta$
\[\hat{\lambda}(\boldsymbol{p}_{1},\beta\boldsymbol{p}_{1},\boldsymbol{q}_{1},\boldsymbol{q}_{2})=\frac{\sum_{i=1}^N\beta^2\, p_{i1} \gamma_i}{\sum_{i=1}^N \beta p_{i1} \gamma_i}= \beta \frac{\sum_{i=1}^N p_{i1} \gamma_i}{\sum_{i=1}^N p_{i1} \gamma_i}= \beta.\]
\end{itemize}

\section{An application of the MPL index to the Italian cultural supply data }\label{sec:empirical}
In this section, we provide an application of the MPL  index to Italian cultural supply data, such as revenues and the number of visitors to museums (i.e. monuments, archeological sites, museum circuits, $\dots$).
The availability of  temporal and  geographical data on Italian culture provides a stimulating basis for ascertaining the potential of the MPL price-index methodology set forth in this paper. The flexibility of the MPL index paves the way to moving beyond ISTAT (and similar) analyses, which are confined to price indexes on the supply of data on Italian culture like access to museums and entertainment sectors, aggregated at the national level \citep{istatr2018b}. 
In addition, to evaluate the performance of the MPL index, we have made a comparison with the CPD/TPD price indexes \citep{diewert2005weighted, rao2004country}, using both real and simulated data. Reference has been made to this approach because, under the log-normality assumption of the error term, the maximum likelihood estimator of the said price index tallies with the least square one, likewise with the MPL index.

As for the nature of the data, note that Italian cultural heritage is at the top of various world-class lists and plays a key role in the Italian economy \citep[see, e.g.,][]{alderighi2018flight}.\footnote{The Italian heritage supply chain accounts for 4,976 museums and the like; it generated almost 200 million euro of revenues in 2017 and employs more than 45,000 people \citep{istatr2016}. }
Lately, local cultural supply has evolved significantly. Indeed, most of the Italian museum circuits were founded relatively recently.\footnote{Approximately 2,300 sites (45.5\%)  of the Italian cultural supply chain were opened between 1960 and 1999, while 2,200 sites (38.6\%) were opened in 2000, taking advantage of the investments for economic recovery and infrastructure enhancement made for Italian cultural heritage sites \citep{istatr2016}.} 
In the following analysis, we have considered the ranking of the top 30 Italian cultural institutions (museums and the like) according to the highest number of annual visitors since 2004 (data source: \href{www.statistica.beniculturali.it}{www.statistica.beniculturali.it}). Among these, only 20 of the internationally renowned institutions remained ranked in the top 30 on a yearly basis. The other 10 positions were held by museums and institutions which only temporarily experienced an outstanding flow of visitors. These observations led to a twofold issue. First, the need to set-up a price index finalized at registering changes in the period under consideration. Second, a dynamic updating method of the index in order to preserve the information associated with the 10 positions not consistently present in the top 30 ranking (aspects which are not considered in the approaches usually adopted by most statistical agencies). 
The multi-period version of the MPL price index proves suitable for this scope. Hence, it has been applied to the data set which collects the number of visitors and revenues of the 30 leading Italian cultural institutions which, from 2004 to 2017, were ranked at least two times in the top 30.\footnote{All analyses in this investigation have been made with our own codes, written in R.}
Figure~\ref{fig:5} shows the MPL price index together with its annual percentage variations for the period 2004--2017: 2004 being the base year and 2017 the year used for updating the index. 
We can note that in the early years of our new Millennium, when important investments started being made in the Italian cultural sector, the prices of museums (and the like) tickets grew (Figure~\ref{fig:5}). Thereafter, the price dynamic became more moderate and then tapered in 2009 and 2014 when, the so called ``W'' recession, namely the international financial and debt crisis in European peripheral countries, hit Italy. 

\begin{figure}[h!]
\begin{center}
\includegraphics[scale=0.42]{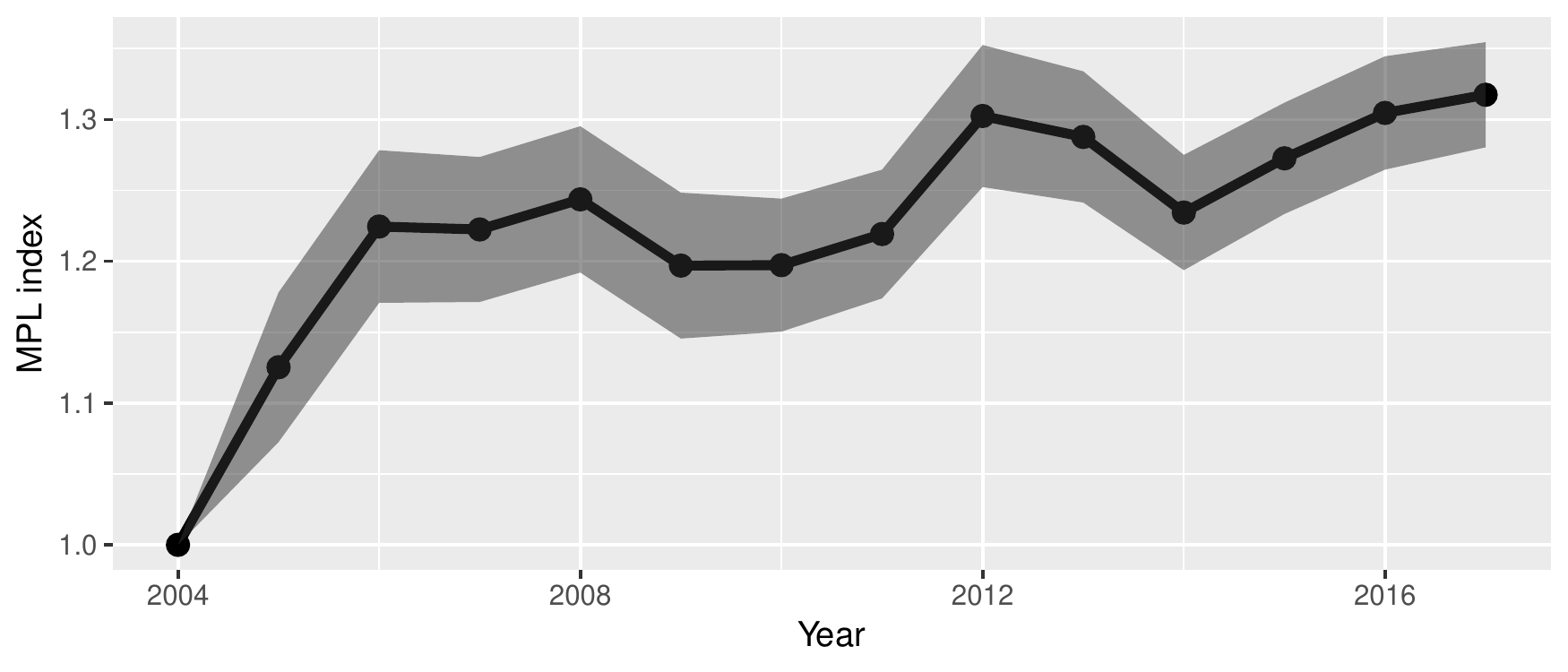}
\includegraphics[scale=0.42]{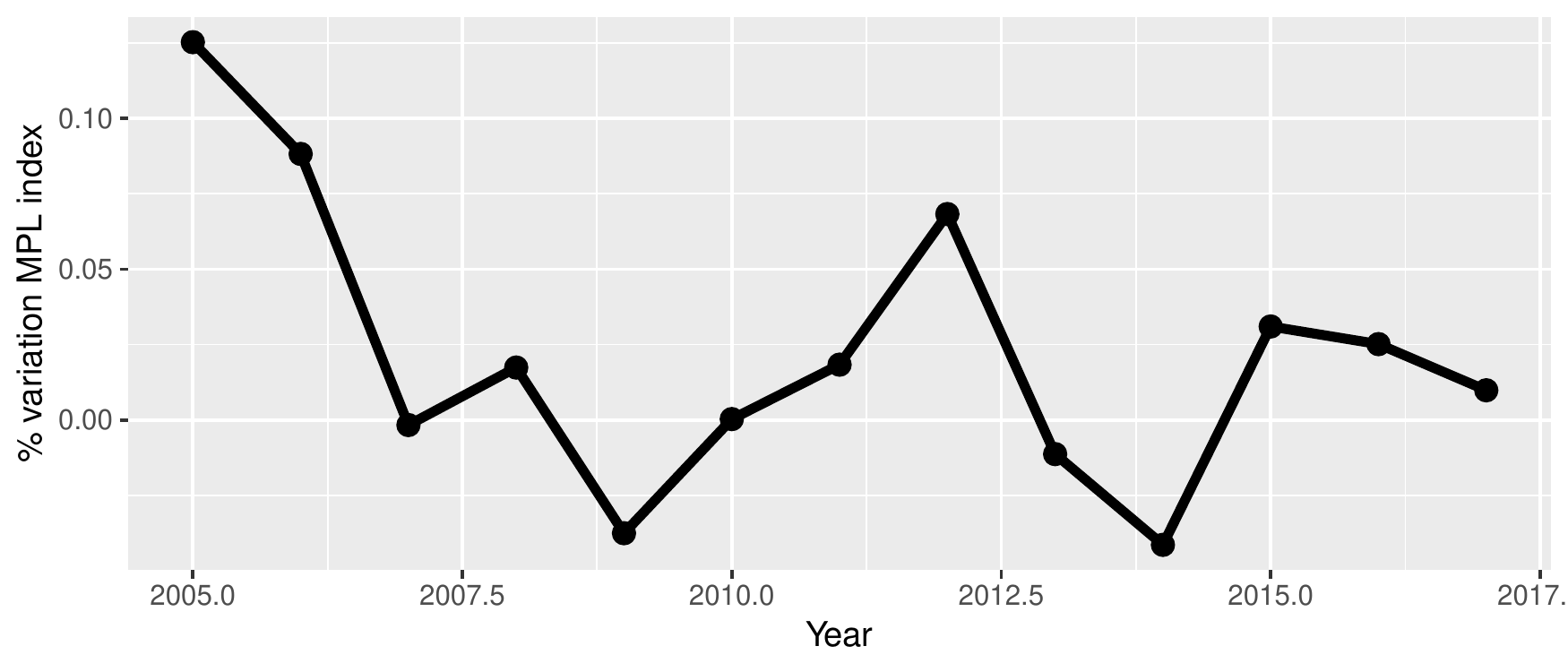}
\caption{MPL index and percentage annual change.}
\label{fig:5}
\end{center}
\end{figure}
\noindent For the sake of further evidence from an empirical standpoint, a comparison of the MPL and the TPD index has been conducted. First TPD price indexes have been computed only for those museums whose prices are available at all times. This has led to a drop in the number of museums/monuments/archeological sites from 36 to 17.  
Figure~\ref{fig:5g} (first panel) shows both the MPL and TPD price indexes together with their 3$\sigma$ confidence bounds. The result is that MPL 
indexes always fall within the confidence bounds of TPD indexes. This result highlights their alignment with the latter, and provides evidence of their greater accuracy, due to their lower standard errors.   
In the case when not all items (museums) are priced in all periods, TPD estimates have been obtained by using the time version of the weighted CPD \citep[][pp. 420-421]{rao2004country}. Figure~\ref{fig:5g} (second panel) shows both MPL and TPD indexes together with their $3\sigma$ confidence bounds. The same comments of the complete price tableau case apply here.   
\begin{figure}[h!]
\begin{center}
\includegraphics[scale=0.42]{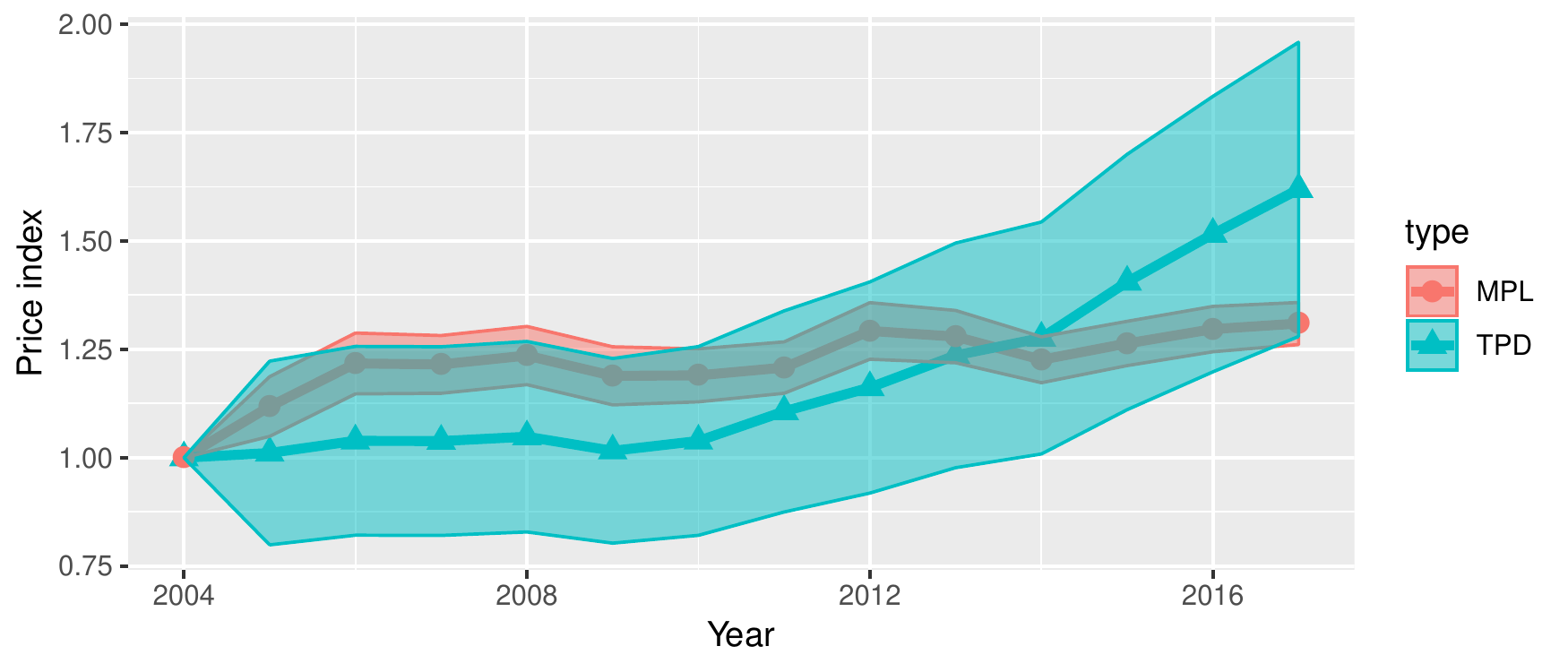} \includegraphics[scale=0.42]{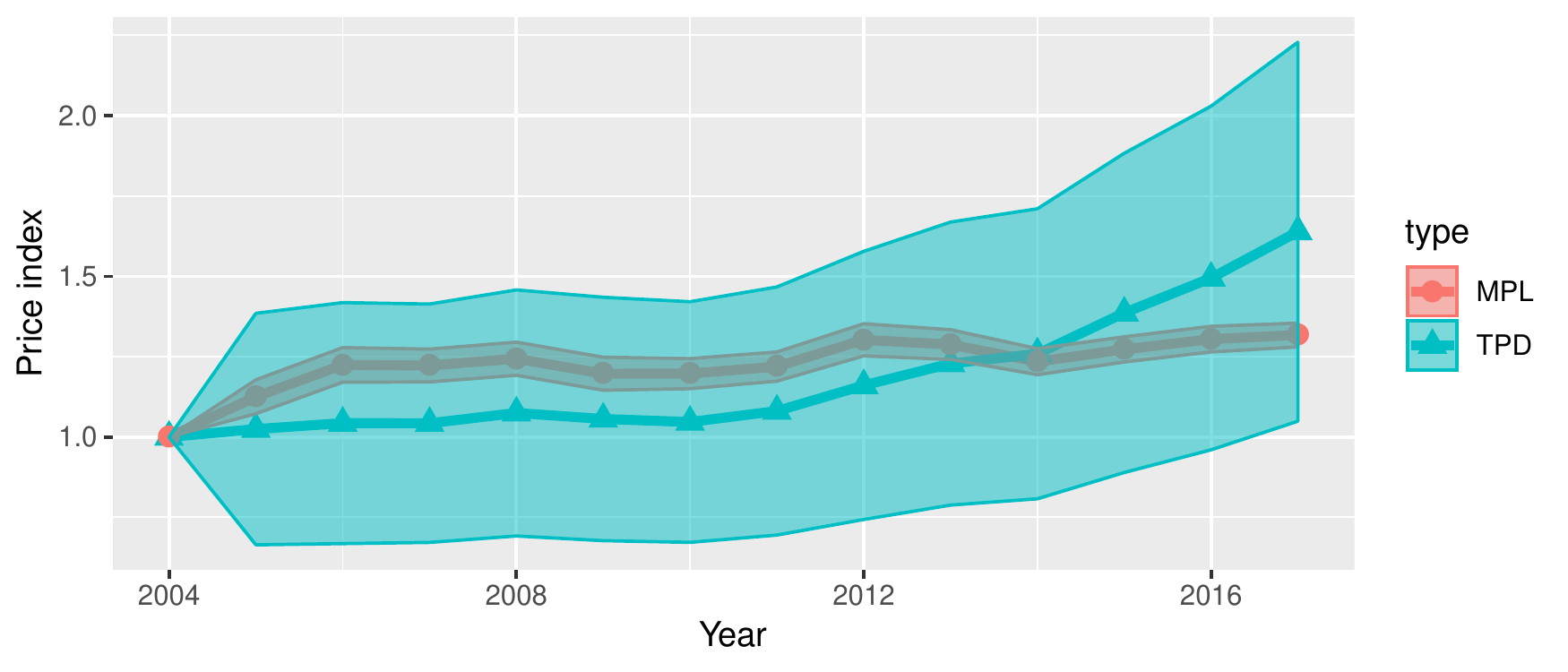}
\caption{In the first panel the MPL index is compared to the TPD index with data from 17 museums always ranked in the top 30; in the second panel the MPL index is compared to the weighted TPD index with data from 36 museums ranked in the top 30 at least twice. 
}
\label{fig:5g}
\end{center}
\end{figure}

\noindent The availability of data on visitors and revenues in 2017 for museums, monuments, archaeological sites, and museum circuits in the North-West, North-East, Centre and South (which includes the two islands Sicily and Sardinia) has allowed the computation of the multilateral version of the MPL index. Looking at the data, we see that almost half (46.3\%) are located in the North, while 28.5\% in the Centre, and 25.2\% in the South and Islands. The Regions with the highest number of cultural institutions are Tuscany (11\%), followed by Emilia-Romagna (9.6\%), Piedmont (8.6\%) and Lombardy (8.2\%)  \citep{istatr2016}. However, alongside the more famous attractions, Italy is home to a wide and rich array of notable locations of cultural interest. A considerable percentage of these places (17.5\%) are found in municipalities with less than 2,000 inhabitants, but which can have up to four or five cultural sites in their small area. Almost a third (30.7\%) are distributed in 1,027 municipalities with a population varying from 2,000 to 10,000, and a bit more than half (51.8\%) are situated in 712 municipalities with a population of 10,000 to 50,000.
Italy is, therefore, characterized by a strongly polycentric cultural supply distributed throughout its territory, even in areas considered as marginal from a geographic stance.
Table~\ref{tab:4b} reports, in the first row, the MPL index computed applying Eq.~\eqref{eq:ffrt0} to the first three areas (North-West, North-East and Centre) considering the Centre as base area. The second row shows the updated values of the MPL index when the South-Islands are added to the data-set. As for the multi-period case, a comparison of the MPL estimates with the CPD ones is provided. The third row of Table~\ref{tab:4b} shows CPD estimates in the case of full price tableau, as all commodities are priced in the four geographic areas. Figure~\ref{fig:5c} shows both the MPL and CPD indexes together with their 3$\sigma$ confidence bounds. Once again, the estimates of the MPL index turn out to be more accurate than those provided by the CPD approach, as the former have standard errors lower than the latter. As in the comparison with the TPD index, the confidence bounds of CPD indexes always include MPL estimates, thus suggesting the compatibility of both indexes.
It is worth noting that in 2017, access to cultural sites in Southern Italy cost the most: almost twice as much as in the North-Eastern area. While the disparity could be ascribed to several factors, such as different costs of managing museums and similar institutions, tourism flows, etc: that type of analysis goes beyond the scope of the current investigation. 
\begin{table}[h!]
\center
\caption{Updated MPL index compared to the CPD index (standard error in parentheses).}
\label{tab:4b}
\begin{tabular}{lllll}
\hline\hline
                                     		& North West & North East  & Centre  & South    \\
\hline
MPL                          	&1.072  (0.185) &0.621  (0.190)   &1.000    &   \\
Updated MPL		&1.070 (0.164) &0.622  (0.170)    &1.000     &1.142 (0.086) \\
CPD 					&1.524 (0.337) &1.283 (0.284)	   &1.000	&  1.021 (0.226)\\
\hline\hline
\end{tabular}
\end{table}

\begin{figure}[h!]
\begin{center}
\includegraphics[scale=0.42]{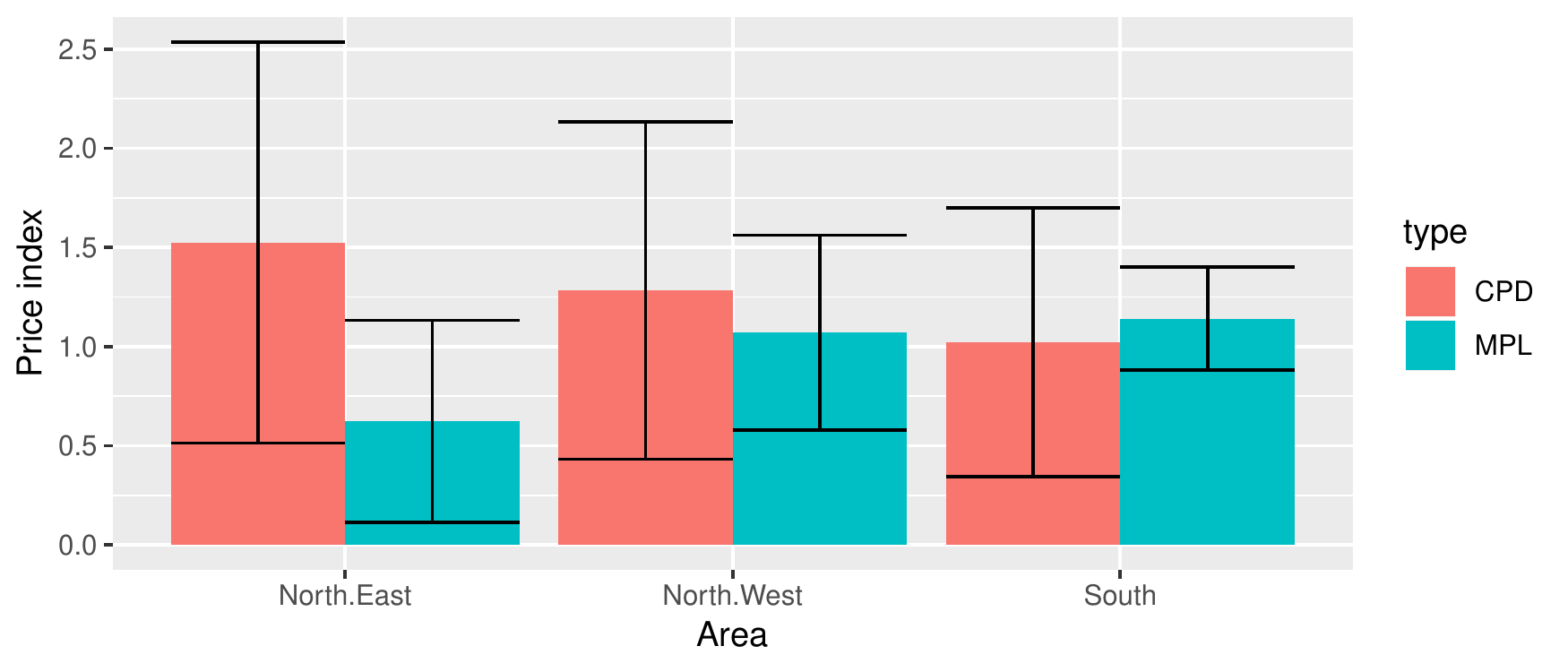}
\caption{MPL and CPD indexes with their 3$\sigma$ confidence bounds.}
\label{fig:5c}
\end{center}
\end{figure}

\noindent Finally, in order to investigate more thoroughly the performance of the MPL index as compared to the TPD one, a simulation analysis has been performed. One thousand simulations were carried out by using perturbed values (and prices as a by-product) and assuming fixed quantities (i.e. equal to the original ones). Next, the simulated values (and prices) were used to compute MPL and TPD indexes in different settings: with and without missing values and/or quantities (and accordingly prices). The final MPL and TPD indexes were obtained as averages of all indexes computed on simulated values and prices. 
Two types of simulations were carried out. First simulated values from the $2^{nd}$ to the $T^{th}$ period (base period values, $\boldsymbol{v}_1$, being kept fixed) were obtained from the original ones ($\boldsymbol{V}_1$) by adding random terms drawn from Normal laws with different means and variances. Plots in Figure~\ref{fig:sim1} show both the MPL and the TPD indexes, for complete and incomplete price tableau cases, together with the associated $3\sigma$ confidence bounds. 

\begin{figure}[h!]
\begin{center}
\includegraphics[scale=0.42]{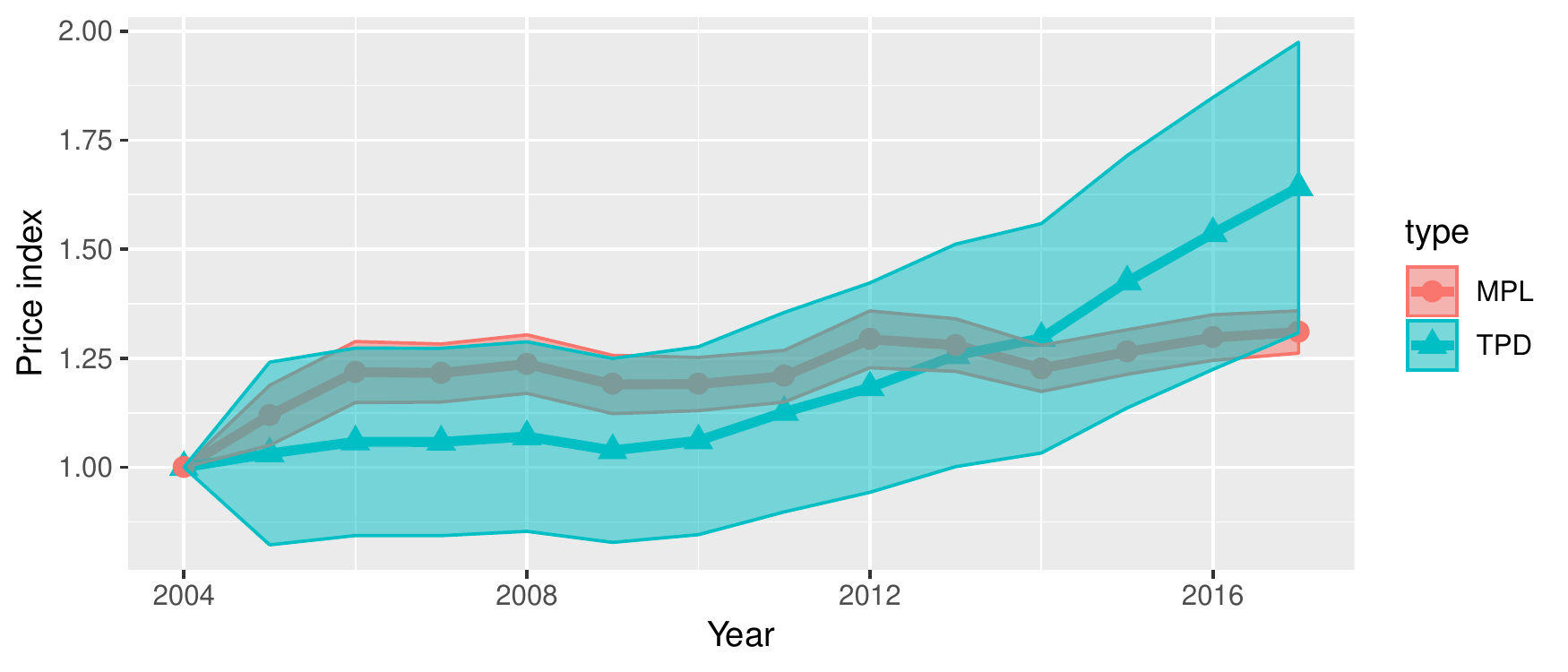} 
\includegraphics[scale=0.42]{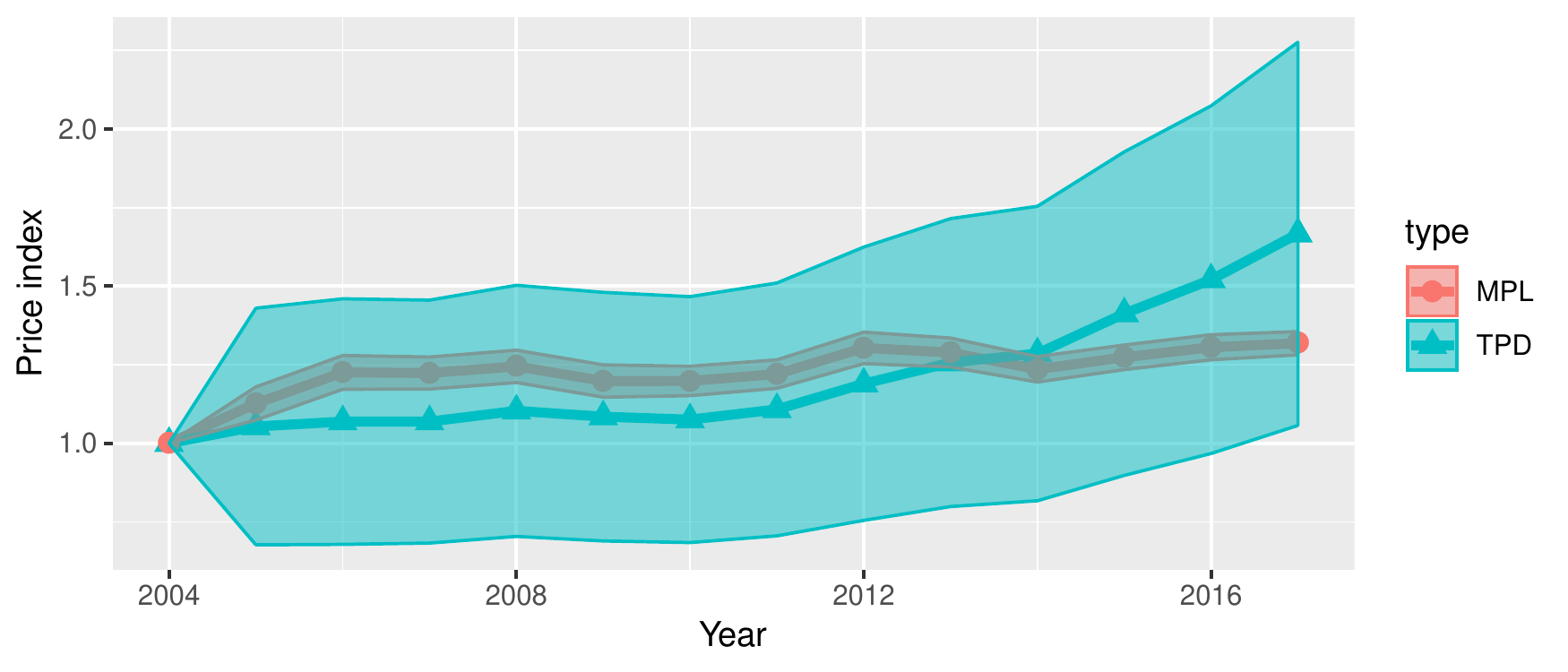}
\caption{MPL and TPD indexes obtained from simulated data generated by adding to $\boldsymbol{V}_1$ random terms drawn from a Normal law with a mean equal to 20000 and a standard error varying randomly from 0 to 1000. The first and second panel respectively refer to a complete and an incomplete price tableau scenario. 
}
\label{fig:sim1}
\end{center}
\end{figure}
Then, simulated values, from the $2^{nd}$ to the $T^{th}$ period (base period values, $\boldsymbol{v}_1$, being kept fixed), were obtained from simulated values of the previous period with the addition of error terms drawn from Normal laws with given means and variances. 
Plots in Figure~\ref{fig:sim2} show both the MPL and the TPD indexes, for the case of complete and incomplete price tableaus, together with the associated $3\sigma$ confidence bounds.  
In both cases, the MPL estimates are in line with the TPD ones, but are more accurate than the latter as their tighter confidence bounds show. 
\begin{figure}[htbp!]
\begin{center}
\includegraphics[scale=0.42]{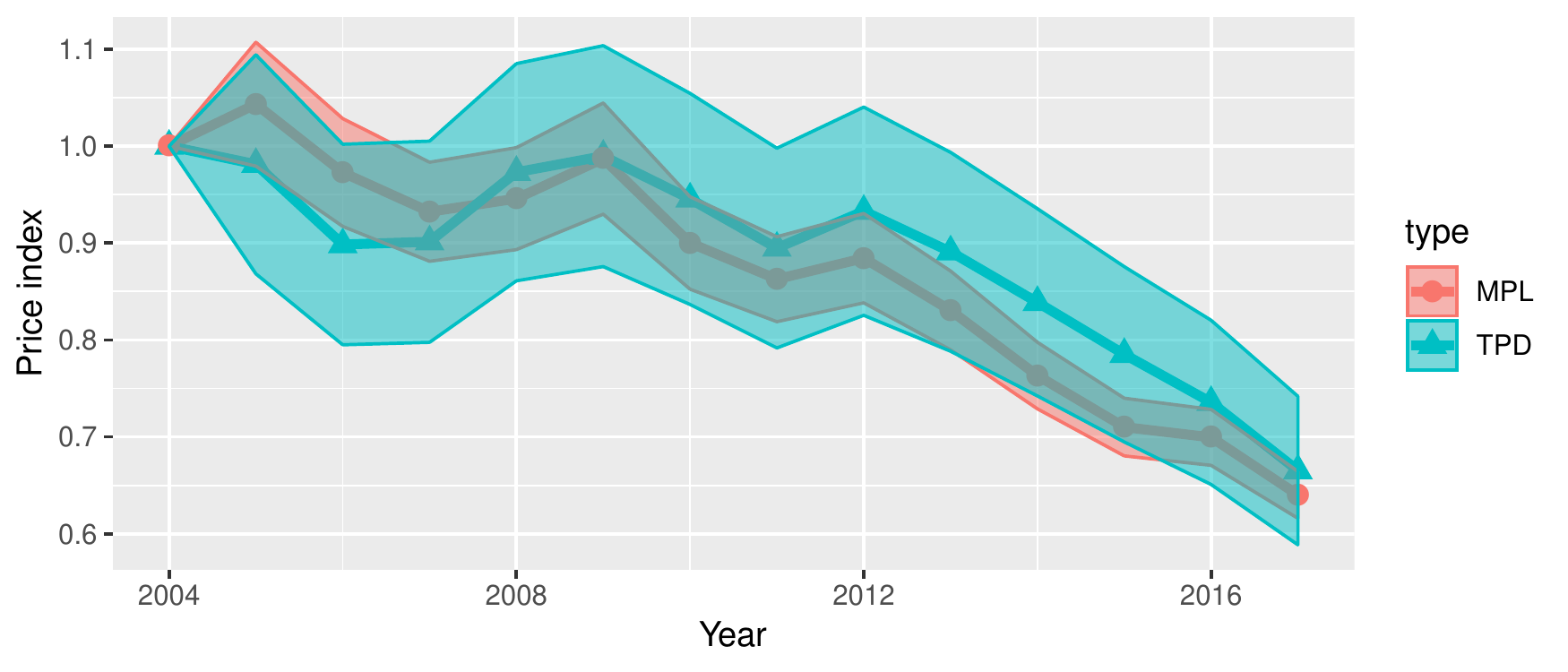} 
\includegraphics[scale=0.42]{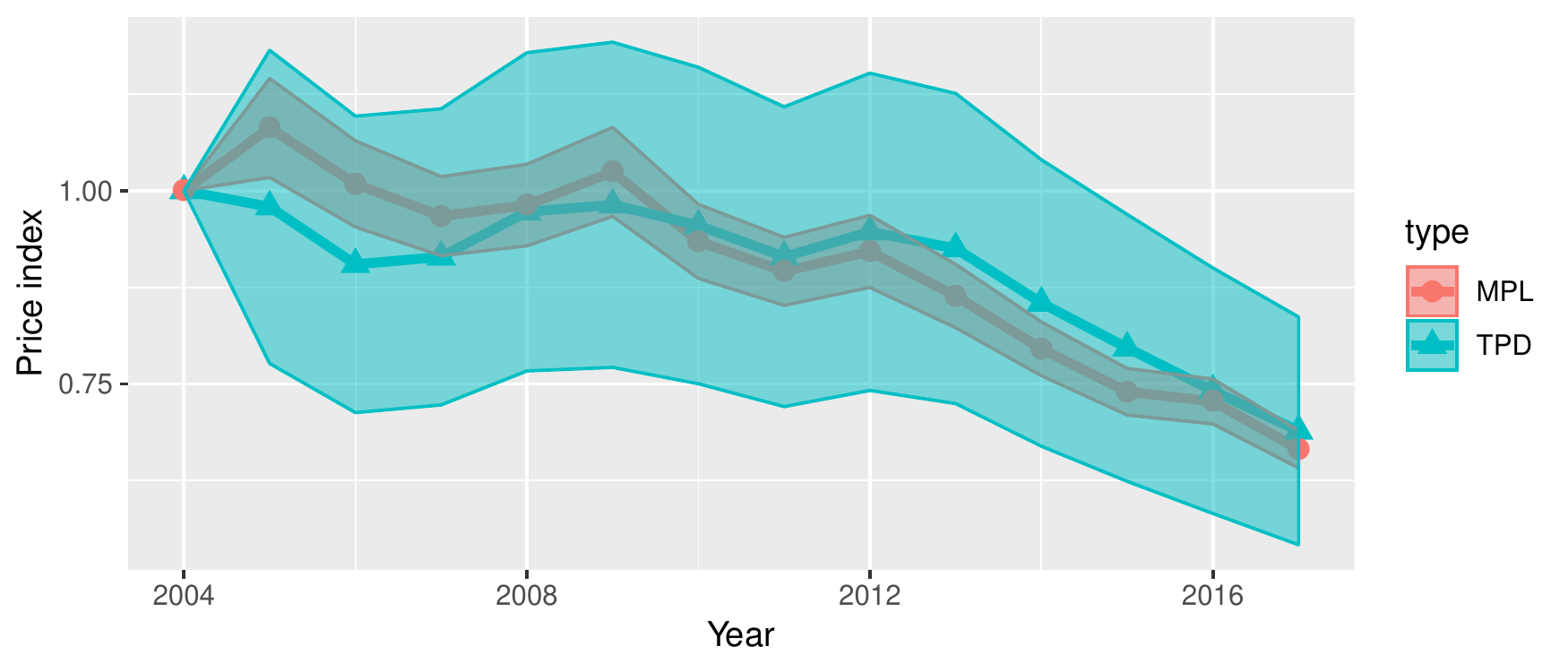}
\caption{MPL and TPD indexes obtained from simulated values $\boldsymbol{\upsilon}_t$ obtained by adding stochastic terms drawn from a Normal law with a mean equal to -5000 and a standard error varying randomly from 0 to 800 to $\boldsymbol{\upsilon}_{t-1}$, for $t=2,...T$. The first and second panel respectively refer to a complete and an incomplete price tableau scenario. }
\label{fig:sim2}
\end{center}
\end{figure}

\section{Conclusion}\label{sec:conclusion}
The paper works out a novel price index that can be used either as a multi-period or as a multilateral index. This index, called MPL index, is obtained as a solution to an ``ad hoc'' minimum-norm criterion, within the framework of the stochastic approach. The computation of the  MPL index does not require the knowledge of commodity prices, but only their quantities and values. 
The  reference basket of the MPL index, over periods or across countries, is more informative and complete than the ones commonly used by statistical agencies, and easy to update. The updating process is twofold depending on the multi-period or the multilateral use of the index. An application of the MPL index to the Italian cultural supply data provides proof of its positive performance. A comparison between the MPL and the CPD/TPD index on both real and simulated data 
provides evidence of the greater accuracy of the MPL estimates. 

%
\clearpage

\bibliographystyle{apalike}
\bibliography{arXiv_MPL_300919}

%

\newpage

\appendix\label{app}

\section*{Appendix}\label{app:proofs}
\setcounter{equation}{0} \renewcommand{\theequation}{A.\arabic{equation}}

\section{Proofs of Theorems and Corollaries}
\subsection{Proof of Theorem~\ref{th:1}}\label{app:proofs1}
\begin{proof}
The optimization problem is
\begin{equation}\label{eq:min}
 \min_{\lambda}\,\,\, \boldsymbol{e}'\boldsymbol{A}\boldsymbol{e}.
 \end{equation}
 \noindent The first order condition for a minimum are
 \[\begin{split}\frac{\partial  \boldsymbol{e}'\boldsymbol{A}\boldsymbol{e}}{\partial \lambda}=0\end{split}\]
which lead to the solution
 \begin{equation}\label{eq:000}\lambda=\frac{\boldsymbol{p}_2'\boldsymbol{A}\boldsymbol{p}_1}{\boldsymbol{p}_1'\boldsymbol{A}\boldsymbol{p}_1}.\end{equation}

\noindent Setting $\boldsymbol{A}=(\boldsymbol{q}_1\,\boldsymbol{q}_1')$ in Eq.~\eqref{eq:000} yields the Laspeyeres index
 \begin{equation}
\lambda_L=\frac{\sum_{i=1}^N p_{i2}q_{i1}}{\sum_{i=1}^N p_{i1}q_{i1}}=\frac{\boldsymbol{p}_2'\boldsymbol{q}_1}{\boldsymbol{p}_1'\boldsymbol{q}_1}.
 \end{equation}
 
\noindent Setting $\boldsymbol{A}=(\boldsymbol{q}_2\,\boldsymbol{q}_2')$ in Eq.~\eqref{eq:000} yields the Paasche index
 \begin{equation}
\lambda_P=\frac{\sum_{i=1}^N p_{i2}q_{i2}}{\sum_{i=1}^N p_{i1}q_{i2}}=\frac{\boldsymbol{p}_2'\boldsymbol{q}_2}{\boldsymbol{p}_1'\boldsymbol{q}_2}.
 \end{equation}

\noindent Setting $\boldsymbol{A}=(\boldsymbol{q}_1+ \boldsymbol{q}_2)\,(\boldsymbol{q}_1+\boldsymbol{q}_2)'$ in  Eq.~\eqref{eq:000} yields the Marshall-Edgeworth index
 \begin{equation}
\lambda_{M-E}=\frac{\sum_{i=1}^N p_{i2}(q_{i1}+q_{i2})}{\sum_{i=1}^N p_{i1}(q_{i1}+q_{i2})}=\frac{\boldsymbol{p}_2'(\boldsymbol{q}_1+\boldsymbol{q}_2)}{\boldsymbol{p}_1'(\boldsymbol{q}_1+\boldsymbol{q}_2)}.
 \end{equation}

\noindent Finally, setting $\boldsymbol{A}=(\tilde{\boldsymbol{q}}\tilde{\boldsymbol{q}}')$ in Eq.~\eqref{eq:000}, where  $\tilde{\boldsymbol{q}}$ is a vector whose elements are the square roots of the entries of the vector $(\boldsymbol{q}_1*\boldsymbol{q}_2)$, leads to the Walsh index
 \begin{equation}
\lambda_{W}=\frac{\sum_{i=1}^N p_{i2}(q_{i1}q_{i2})^{\frac{1}{2}}}{\sum_{i=1}^N p_{i1}(q_{i1}+q_{i2})^{\frac{1}{2}}}=\frac{\boldsymbol{p}_2'\tilde{\boldsymbol{q}}}{\boldsymbol{p}_1'\tilde{\boldsymbol{q}}}.
 \end{equation}
 \end{proof}

\subsection{Proof of Theorem~\ref{th:2}}\label{app:proofs2}

\begin{proof}
In compact form, the model in Eq.~\eqref{eq:sist1} can be written as
\begin{equation}\label{eq:reg}
\underset{(NT, 1)}{\boldsymbol{y}}=\underset{(NT, N+T-1)}{\boldsymbol{X}_{}}\underset{(N+T-1, 1)}{\boldsymbol{\beta}} + \underset{(NT, 1)}{\boldsymbol{\mu}}
\end{equation} 
where 
\begin{equation}\label{eq:reg1}
\underset{(NT,1)}{\boldsymbol{y}}=\begin{bmatrix} \underset{(N,1)}{\boldsymbol{v}_1}\\\underset{(N(T-1),1)}{\boldsymbol{0}}
\end{bmatrix}, \,\,\underset{(NT,N+T-1)}{\boldsymbol{X}_{}}=\begin{bmatrix} \underset{(N,T-1)}{\boldsymbol{0}} &  \underset{(N,N)}{(\boldsymbol{q}'_1\otimes \boldsymbol{I}_N)\boldsymbol{R}_N'}\\ \underset{(N(T-1),T-1)}{(\boldsymbol{I}_{T-1}\otimes (-\boldsymbol{V}_1))\boldsymbol{R}_{T-1}'} & \underset{(N(T-1),N)}{(\boldsymbol{Q}'_1\otimes \boldsymbol{I}_N)\boldsymbol{R}_N'}
\end{bmatrix}
\end{equation}
and
\[
\underset{(N+T-1,1)}{\boldsymbol{\beta}}=\begin{bmatrix} \underset{(T-1,1)}{{\boldsymbol{\delta}}}\\\underset{(N,1)}{\tilde{\boldsymbol{p}}}
\end{bmatrix}, \,\,\underset{(NT,1)}{\boldsymbol{\mu}}=\begin{bmatrix} \underset{(N,1)}{\boldsymbol{\varepsilon}_1} \\ \underset{(N(T-1),1)}{\boldsymbol{\eta}}
\end{bmatrix}.
\]

\noindent The ordinary least square estimator of the vector $\boldsymbol{\beta}$ is given by
\begin{equation}
\label{eq:beta}
\hat{\boldsymbol{\beta}}=(\boldsymbol{X}'\boldsymbol{X})^{-1}\boldsymbol{X}'\boldsymbol{y}
\end{equation}
where
\begin{equation}
\label{eq:xp}
\begin{split}
\underset{(N+T-1,N+T-1)}{(\boldsymbol{X}'\boldsymbol{X})}&=
\begin{bmatrix} \boldsymbol{R}_{T-1}(\boldsymbol{I}_{T-1}\otimes\boldsymbol{V}_1'\boldsymbol{V}_1) \boldsymbol{R}_{T-1}' & \boldsymbol{R}_{T-1}(\boldsymbol{Q}_1'\otimes (-\boldsymbol{V}_1'))\boldsymbol{R}'_{N}
\\ \boldsymbol{R}_{N}(\boldsymbol{Q}_1\otimes (-\boldsymbol{V}_1))\boldsymbol{R}_{T-1}' & \boldsymbol{R}_{N} ((\boldsymbol{q}_1\boldsymbol{q}_1'+\boldsymbol{Q}_1\boldsymbol{Q}_1')\otimes \boldsymbol{I}_N)\boldsymbol{R}_N'
\end{bmatrix}
\\
&=\begin{bmatrix}\boldsymbol{I}_{T-1}*\boldsymbol{V}_1'\boldsymbol{V}_1 & -\boldsymbol{Q}_1'*\boldsymbol{V}_1'\\
-\boldsymbol{Q}_1*\boldsymbol{V}_1 &  (\boldsymbol{q}_1\boldsymbol{q}_1'+\boldsymbol{Q}_1\boldsymbol{Q}_1')* \boldsymbol{I}_N
\end{bmatrix}
\end{split}
\end{equation}
and
\begin{equation}
\label{eq:xy1}
\underset{(N+T-1,1)}{\boldsymbol{X}'\boldsymbol{y}}=
\begin{bmatrix}
\underset{(T-1,1)}{\boldsymbol{0}}\\\underset{(N,1)}{\boldsymbol{R}_N(\boldsymbol{q}_1\otimes \boldsymbol{I}_N)\boldsymbol{v}_1}\end{bmatrix}\footnote{Note that 
\[\boldsymbol{R}_N(\boldsymbol{q}_1\otimes \boldsymbol{I}_N)\boldsymbol{v}_1=\boldsymbol{R}_N(\boldsymbol{q}_1\otimes \boldsymbol{I}_N)(\boldsymbol{I}_1\otimes\boldsymbol{v}_1)\boldsymbol{R}_1'=\boldsymbol{R}_N(\boldsymbol{q}_1\otimes \boldsymbol{v}_1)\boldsymbol{R}_1'=\boldsymbol{q}_1* \boldsymbol{v}_1
 \,\,\,\mbox{where} \,\,\,\boldsymbol{R}_1=\underset{(1,1)}{\boldsymbol{e}'_1}\otimes \underset{(1,1)}{\boldsymbol{e}'_1}=1\].
}
= 
\begin{bmatrix}
\underset{(T-1,1)}{\boldsymbol{0}}\\\underset{(N,1)}{\boldsymbol{q}_1 *\boldsymbol{v}_1}
\end{bmatrix}.\footnote{Use has been made of the following relationship between the Kronecker and the Hadamard product \citep{faliva1996hadamard} \[\underset{(N,M)}{\boldsymbol{A}} * \underset{(N,M)}{\boldsymbol{B}}=\boldsymbol{R}_N(\boldsymbol{A}\otimes\boldsymbol{B})\boldsymbol{R}'_M\] to obtain the right-hand sides of Eq.~\eqref{eq:xp} and \eqref{eq:xy1}.}
\end{equation}
Accordingly,
\begin{equation}
\hat{{\boldsymbol{\delta}}}=\begin{bmatrix}\boldsymbol{I}_{T-1}\,\,\,\underset{(T-1,N)}{\boldsymbol{0}}
\end{bmatrix}\underset{(N+T-1,1)}{\hat{\boldsymbol{\beta}}}=
\boldsymbol{\Lambda}_{12}(\boldsymbol{q}_1* \boldsymbol{v}_1)
\label{eq:deltt}
\end{equation}
where $\boldsymbol{\Lambda}_{12}$  is the upper off diagonal block of the inverse matrix
\[\underset{(N+T-1,N+T-1)}{(\boldsymbol{X}'\boldsymbol{X})^{-1}}
=\boldsymbol{\Lambda}= \begin{bmatrix}\underset{(T-1,T-1)}{\boldsymbol{\Lambda}_{11}}&\underset{(T-1,N)}{\boldsymbol{\Lambda}_{12}}\\
\underset{(N,T-1)}{\boldsymbol{\Lambda}_{21}} & \underset{(N,N)}{\boldsymbol{\Lambda}_{22}}
\end{bmatrix}.
\]
Partitioned inversion \citep[see][]{faliva2008dynamic} leads to 
\begin{footnotesize}\[\begin{split}\boldsymbol{\Lambda}_{12}=\left\{
(\boldsymbol{I}_{T-1}*\boldsymbol{V}_1'\boldsymbol{V}_1)-(\boldsymbol{Q}_1'*\boldsymbol{V}_1')\left[
(\boldsymbol{q}_1\boldsymbol{q}_1'+\boldsymbol{Q}_1\boldsymbol{Q}_1')*\boldsymbol{I}_N
\right]^{-1}(\boldsymbol{Q}_1*\boldsymbol{V}_1)
\right\}^{-1}(\boldsymbol{Q}_1'*\boldsymbol{V}_1')\left[
(\boldsymbol{q}_1\boldsymbol{q}_1'+\boldsymbol{Q}_1\boldsymbol{Q}_1')*\boldsymbol{I}_N
\right]^{-1}
\end{split}\]\end{footnotesize}
and this yields Eq.~\eqref{eq:delta}.
\end{proof}
%
%
\subsection{Proof of Corollary~\ref{cor:1}}\label{app:proofs3}
\begin{proof} When $T=2$, $\boldsymbol{Q}_1=\underset{(N,1)}{\boldsymbol{q}_2}$ and $\boldsymbol{V}_1=\underset{(N,1)}{\boldsymbol{v}_2}$. Accordingly, the following holds
\begin{equation}
(\boldsymbol{I}_{T-1}*\boldsymbol{V}_1'\boldsymbol{V}_1)=\boldsymbol{v}_2'\boldsymbol{v}_2.
\end{equation}
\noindent Then, upon noting that
\[
\left[
(\boldsymbol{q}_1\boldsymbol{q}_1'+\boldsymbol{q}_2\boldsymbol{q}_2')*\boldsymbol{I}_N
\right]^{-1}=\begin{bmatrix} \frac{1}{q_{11}^2+q_{12}^2} & 0 & \dots &0 \\ 0 & \frac{1}{q_{21}^2+q_{22}^2}  & \dots &0
\\ \vdots & \vdots & \ddots &\vdots\\ 0 & 0 & \dots &\frac{1}{q_{N1}^2+q_{N2}^2} \end{bmatrix}
\]
where $q_{it}$ denotes the quantity of the $i^{th}$ good at time $t$, some computations prove that
\[\begin{split}\hat{{\delta}}=&\left\{(
\boldsymbol{v}_2'\boldsymbol{v}_2)-(\boldsymbol{q}_2'*\boldsymbol{v}_2')\left[
(\boldsymbol{q}_1\boldsymbol{q}_1'+\boldsymbol{q}_2\boldsymbol{q}_2')*\boldsymbol{I}_N
\right]^{-1}(\boldsymbol{q}_2*\boldsymbol{v}_2)
\right\}^{-1}(\boldsymbol{q}_2'*\boldsymbol{v}_2')\left[
(\boldsymbol{q}_1\boldsymbol{q}_1'+\boldsymbol{q}_2\boldsymbol{q}_2')*\boldsymbol{I}_N
\right]^{-1}(\boldsymbol{q}_1* \boldsymbol{v}_1)\\=&
\left\{
\boldsymbol{v}_2'\boldsymbol{v}_2-[q_{12}v_{12},\,\cdots,\, q_{N2}v_{N2}] 
\begin{bmatrix} \frac{1}{q_{11}^2+q_{12}^2} & 0 & \dots &0 \\ 0 & \frac{1}{q_{21}^2+q_{22}^2}  & \dots &0
\\ \vdots & \vdots & \ddots &\vdots\\ 0 & 0 & \dots &\frac{1}{q_{N1}^2+q_{N2}^2} \end{bmatrix}
\begin{bmatrix}
q_{12}v_{12}\\
\vdots\\
q_{N2}v_{N2}
\end{bmatrix}
\right\}^{-1}
\\& [q_{12}v_{12},\,\cdots,\, q_{N2}f_{N2}] 
\begin{bmatrix} \frac{1}{q_{11}^2+q_{12}^2} & 0 & \dots &0 \\ 0 & \frac{1}{q_{21}^2+q_{22}^2}  & \dots &0
\\ \vdots & \vdots & \ddots &\vdots\\ 0 & 0 & \dots &\frac{1}{q_{N1}^2+q_{N2}^2} \end{bmatrix}
\begin{bmatrix}
q_{11}v_{11}\\
\vdots\\
q_{N1}v_{N1}
\end{bmatrix}=\\
=&
\left\{
\sum_{i=1}^Nv^2_{i2}-\sum_{i=1}^N\frac{q^2_{i2}v^2_{i2}}{q^2_{i1}+q^2_{i2}}
\right\}^{-1}
\left\{
\sum_{i=1}^N\frac{q_{i2}v_{i2}q_{i1}v_{i1}}{q^2_{i1}+q^2_{i2}}
\right\}=\left\{\sum_{i=1}^N\frac{q^2_{i1}v^2_{i2}}{q^2_{i1}+q^2_{i2}}
\right\}^{-1}
\left\{
\sum_{i=1}^N\frac{q_{i2}v_{i2}q_{i1}v_{i1}}{q^2_{i1}+q^2_{i2}}
\right\}.
\end{split}\]
The reciprocal of the deflator $\hat{{\delta}}$ yields the intended price index 
\[
\hat{{\lambda}}=\left\{\sum_{i=1}^N\frac{q^2_{i1}v^2_{i2}}{q^2_{i1}+q^2_{i2}}
\right\}
\left\{
\sum_{i=1}^N\frac{q_{i2}v_{i2}q_{i1}v_{i1}}{q^2_{i1}+q^2_{i2}}
\right\}^{-1}=\frac{\sum_{i=1}^Np_{i2}\pi_i}{\sum_{i=1}^Np_{i1}\pi_i}
\] where 
\begin{equation}
\label{eq:ttp}
\pi_i=p_{i2}\frac{q_{i1}^2q_{i2}^2}{q_{i1}^2+q_{i2}^2}.
\end{equation}
\noindent Multiplying and dividing  the numerator of Eq.~\eqref{eq:ttp} by $p_{i1}$, the index can be also written as 
\begin{equation}
\sum_{i=1}^N\frac{ p_{i2}}{p_{i1}} \tilde{\pi}_{i}
\end{equation}
that is as a convex linear  combination of prices with weights given by 
\begin{equation}
\tilde{\pi}_{i}=\frac{v_{i1}v_{i2}\,\frac{q_{i1}q_{i2}}{q_{i1}^2+q_{i2}^2}}{\sum_{i=1}^N v_{i1}v_{i2}\,\frac{q_{i1}q_{i2}}{q_{i1}^2+q_{i2}^2}}.
\end{equation}

\noindent Further, the index $\hat{{\lambda}}$ can be also rewritten in a compact form as follows
\begin{equation}
\label{eq:comp}
\hat{{\lambda}}=\frac{\left(\boldsymbol q_1\ast \boldsymbol v_2\right)'\boldsymbol D^{-1}\left(\boldsymbol q_1\ast \boldsymbol v_2\right)}{\left(\boldsymbol q_2\ast \boldsymbol v_2\right)'\boldsymbol D^{-1}\left(\boldsymbol q_1\ast \boldsymbol v_1\right)}
\end{equation}
where
\begin{equation}
\boldsymbol D=\begin{bmatrix} q_{11}^2+q_{12}^2 & 0 & \dots &0 \\ 0 & q_{21}^2+q_{22}^2  & \dots &0
\\ \vdots & \vdots & \ddots &\vdots\\ 0 & 0 & \dots &q_{N1}^2+q_{N2}^2 \end{bmatrix}.
\end{equation}
By setting $\boldsymbol D^{-1/2}\boldsymbol q_t=\boldsymbol \tilde{\boldsymbol{q}_t} $ and by making use of the properties of the Hadamard product\footnote{Let $\boldsymbol D$ be a diagonal matrix. Then, simple computations show that $\boldsymbol D\left(\boldsymbol {a}\ast \boldsymbol{b}\right)= \left(\boldsymbol D \boldsymbol {a} \ast \boldsymbol{b}\right)$.}, the intended price index may be eventually written as follows
\begin{equation}
\label{eq:comp1}
\hat{{\lambda}}=\frac{\left(\tilde{\boldsymbol q}_1\ast \boldsymbol v_2\right)'\left(\tilde{\boldsymbol q}_1\ast \boldsymbol v_2\right)}{\left(\tilde{\boldsymbol q}_2\ast \boldsymbol v_2\right)'\left(\tilde{\boldsymbol q}_1\ast \boldsymbol v_1\right)}.
\end{equation}
\end{proof}

\subsection{Proof of Corollary~\ref{cor:3}}\label{app:proofs6} 
\begin{proof}
The variance-covariance matrix of the estimator $\boldsymbol {\widehat{\beta}}$, given in Eq.~\eqref{eq:beta}, is 
\begin{equation}
\label{eq:var}
Var(\boldsymbol {\widehat{\beta}})= \sigma^2[\boldsymbol{X}'\boldsymbol{X}]^{-1}
\end{equation}
where $(\sigma^2\boldsymbol{I}_{NT})$ is the variance-covariance matrix of the stochastic vector $\boldsymbol{\mu}$ given in Eq.~\eqref{eq:reg}.
According to Eq.~\eqref{eq:var}, the variance-covariance matrix of the vector $\boldsymbol{\hat {\delta}}$, as in Eq.~\eqref{eq:deltt}, turns out to be 
\begin{equation}
\label{eq:var1}
V(\boldsymbol {\hat{\delta}})=\sigma^2 [\boldsymbol{I}_{T-1}, \boldsymbol{0}_{T-1,N}][\boldsymbol{X}'\boldsymbol{X}]^{-1}\begin{bmatrix}
\boldsymbol {I}_{T-1}\\
\boldsymbol{0}_{N,T-1}
\end{bmatrix}
\end{equation}
which, with some computations and taking into account Eq.~\eqref{eq:xp}, can be worked out as follows
\begin{equation}
\label{eq:var2}
V(\boldsymbol{\hat{\delta}})=\sigma^2[\boldsymbol{I}_{T-1}*\boldsymbol{V}_1'\boldsymbol{V}_1]^{-1}. 
\end{equation}

\noindent The $t^{th}$ diagonal entries of the above matrix is
\begin{equation}
\label{eq:var6}
var({\hat{\delta}_{t,t}})=\sigma^2\left[\boldsymbol{v}_t'\boldsymbol{v}_t\right]^{-1},
\end{equation}
\noindent where $\boldsymbol {v}_t$ denotes the $t^{th}$ column of the matrix $\boldsymbol{V}_1$.   
\end{proof}

\subsection{Proof of Theorem~\ref{th:4}}\label{app:proofs7} 
\begin{proof}
When the values, $\boldsymbol{v}_{T+1}$, and the quantities, $\boldsymbol{q}_{T+1}$, of  $N$ commodities in a reference basket become available for the $(T+1)^{th}$ additional country, the reference equation system for updating the MPL index becomes

\begin{equation}
\begin{cases}
&\boldsymbol{V}_{}\, \, \boldsymbol{D}_{\boldsymbol{\delta}}=  \boldsymbol{D}_{\tilde{\boldsymbol{p}}}\, \, \boldsymbol{Q}_{} + \boldsymbol{E}_{}
\\& \boldsymbol{v}_{T+1}{\delta}_{T+1}=\boldsymbol{D}_{\tilde{\boldsymbol{p}}}\boldsymbol{q}_{T+1}+\boldsymbol{\varepsilon}_{T+1}
\end{cases}
\rightarrow
\begin{cases}
& \boldsymbol{v}_1=\boldsymbol{D}_{\tilde{\boldsymbol{p}}}\boldsymbol{q}_1+\boldsymbol{\varepsilon}_1
\\& \boldsymbol{V}_1 \tilde{\boldsymbol{D}}_{\boldsymbol{\delta}}=\boldsymbol{D}_{\tilde{\boldsymbol{p}}}\boldsymbol{Q}_1+\boldsymbol{E}_1
\\& \boldsymbol{v}_{T+1}{\delta}_{T+1}=\boldsymbol{D}_{\tilde{\boldsymbol{p}}}\boldsymbol{q}_{T+1}+\boldsymbol{\varepsilon}_{T+1}
\end{cases}.
\end{equation}
After some computations, the above system can be also written as
\begin{equation}
\begin{cases}
&  \boldsymbol{v}_1=(\boldsymbol{q}_1' \otimes \boldsymbol{I}_N) \boldsymbol{R}_N' \tilde{\boldsymbol{p}}+\boldsymbol{\varepsilon}_1
\\
& \underset{(N(T-1),1)}{\boldsymbol{0}}=(\boldsymbol{I}_{T-1}\otimes (\boldsymbol{-V}_1))\boldsymbol{R}_{T-1}'\boldsymbol{\delta}+(\boldsymbol{Q}_1'\otimes \boldsymbol{I}_N)\boldsymbol{R}_N'\tilde{\boldsymbol{p}}+\boldsymbol{\eta}\\
& \boldsymbol{0}=-\boldsymbol{v}_{T+1}\delta_{T+1}+ (\boldsymbol{q}_{T+1}' \otimes \boldsymbol{I}_N) \boldsymbol{R}_N' \tilde{\boldsymbol{p}}+\boldsymbol{\varepsilon}_{T+1}
\end{cases}
\end{equation}
or, in compact form, as
\begin{equation}\label{eq:regu}
\underset{(NT+N, 1)}{\boldsymbol{y}_u}=\underset{(NT+N, N+T)}{\boldsymbol{X}_u}\underset{(N+T, 1)}{\boldsymbol{\beta}_u} + \underset{(NT+N, 1)}{\boldsymbol{\mu}_u}
\end{equation} where 
\begin{equation}\label{eq:reg1}
\underset{(NT+N,1)}{\boldsymbol{y}_u}=\begin{bmatrix} \underset{(N,1)}{\boldsymbol{v}_1}\\\underset{(N(T-1),1)}{\boldsymbol{0}}\\\underset{(N,1)}{\boldsymbol{0}}
\end{bmatrix}, \,\,\underset{(NT+N,N+T)}{\boldsymbol{X}_u}
=\begin{bmatrix} 
\underset{(N,T-1)}{\boldsymbol{0}} &  \underset{(N,1)}{\boldsymbol{0}}&\underset{(N,N)}{(\boldsymbol{q}_1'\otimes \boldsymbol{I}_N)\boldsymbol{R}_N'}\\ 
\underset{(N(T-1),T-1)}{(\boldsymbol{I}_{T}\otimes - \boldsymbol{V}_1)\boldsymbol{R}_{T-1}'} & \underset{(N(T-1),1)}{\boldsymbol{0}}&\underset{(N(T-1),N)}{(\boldsymbol{Q}_1'\otimes \boldsymbol{I}_N)\boldsymbol{R}_N'}\\ 
\underset{(N,T-1)}{\boldsymbol{0}}&\underset{(N,1)}{-\boldsymbol{v}_{T+1}}  & \underset{(N,N)}{(\boldsymbol{q}_{T+1}'\otimes \boldsymbol{I}_N)\boldsymbol{R}_N'}
\end{bmatrix}
\end{equation}
and
\[
\underset{(N+T,1)}{\boldsymbol{\beta}_u}=\begin{bmatrix} \underset{(T-1,1)}{{\boldsymbol{\delta}}}\\\underset{(1,1)}{{\delta_{T+1}}}\\\underset{(N,1)}{\tilde{\boldsymbol{p}}}
\end{bmatrix}, \,\,\underset{(NT+N,1)}{\boldsymbol{\mu}_u}=\begin{bmatrix} \underset{(N,1)}{\boldsymbol{\varepsilon}_1} \\ \underset{(N(T-1),1)}{\boldsymbol{\eta}}\\ \underset{(N,1)}{\boldsymbol{\varepsilon}_{T+1}}
\end{bmatrix}.
\]

\noindent Then, following the same argument of Theorem~\ref{th:2}, we obtain that

\begin{equation}
\underset{(N+T,N+T)}{\boldsymbol{X}_u'\boldsymbol{X}_u}
=\begin{bmatrix} 
\boldsymbol{I}_{T-1}*\boldsymbol{V}_1\boldsymbol{V}_1' & \boldsymbol{0} &-\boldsymbol{Q}_1'*\boldsymbol{V}_1'
\\ 
\boldsymbol{0}&\boldsymbol{v}_{T+1}'\boldsymbol{v}_{T+1}&-\boldsymbol{v}_{T+1}'*\boldsymbol{q}_{T+1}'\\
-\boldsymbol{Q}_1*\boldsymbol{V}_1&-\boldsymbol{v}_{T+1}*\boldsymbol{q}_{T+1}&({\boldsymbol{q}}_1{\boldsymbol{q}}_1'+{\boldsymbol{Q}}_1{\boldsymbol{Q}}_1'+\boldsymbol{q}_{T+1}\boldsymbol{q}_{T+1}')*\boldsymbol{I}_N
\end{bmatrix}
\end{equation}
and
\begin{equation}
\underset{(N+T,1)}{\boldsymbol{X}_u'\boldsymbol{y}_u}
=\begin{bmatrix} 
\underset{(T,1)}{\boldsymbol{0}}\\\underset{(N,1)}{\boldsymbol{q}_1*\boldsymbol{v}_1}
\end{bmatrix}.
\end{equation}
Upon nothing that 
\begin{equation}
\begin{bmatrix}
\underset{(T-1,1)}{\hat{\boldsymbol{\delta}}}
\\\underset{(1,1)}{\hat{\delta}_{T+1}}
\end{bmatrix}=
\begin{bmatrix}
\boldsymbol{I}_T&\underset{(T,N)}{\boldsymbol{0}}
\end{bmatrix} \underset{(N+T,1)}{\boldsymbol{\hat{\beta}_u}} = \boldsymbol{\Lambda_{12}}\left(\boldsymbol{q}_1\ast\boldsymbol{v}_1\right)
\end{equation}
where $\boldsymbol{\Lambda}_{12}$ is the upper off diagonal block of the inverse matrix $(\boldsymbol{X}_u'\boldsymbol{X}_u)^{-1}$
\[\underset{(N+T,N+T)}{(\boldsymbol{X}_u'\boldsymbol{X}_u)^{-1}}
=\boldsymbol{\Lambda}=\begin{bmatrix}\underset{(T,T)}{\boldsymbol{\Lambda}_{11}}&\underset{(T,N)}{\boldsymbol{\Lambda}_{12}}\\
\underset{(N,T)}{\boldsymbol{\Lambda}_{21}} & \underset{(N,N)}{\boldsymbol{\Lambda}_{22}}
\end{bmatrix},
\]
partitioned inversion leads to 
\begin{footnotesize}
\begin{equation}
\begin{split}
\underset{(T,N)}{\boldsymbol{\Lambda}_{12}}=&\\
=&\left\{
\begin{bmatrix} 
\boldsymbol{I}_{T-1}\ast\boldsymbol{V}'_1\boldsymbol{V}_1 & \boldsymbol{0}  \\ 
\boldsymbol{0} & \boldsymbol{v}'_{T+1}\boldsymbol{v}_{T+1}  \end{bmatrix}-
\begin{bmatrix} 
\boldsymbol{Q}'_1\ast\boldsymbol{V}'_1  \\ 
\boldsymbol{v}'_{T+1}\ast\boldsymbol{q}'_{T+1}  \end{bmatrix}
\begin{bmatrix} 
\left(\boldsymbol{q}_1\boldsymbol{q}'_1+\boldsymbol{Q}_1\boldsymbol{Q}'_1 +\boldsymbol{q}_{T+1}\boldsymbol{q}'_{T+1}\right) \ast\boldsymbol{I}_N\end{bmatrix}^{-1} \right.
\\&
\left.\begin{bmatrix} 
\boldsymbol{Q}_1\ast\boldsymbol{V}_1  & 
\boldsymbol{v}_{T+1}\ast\boldsymbol{q}_{T+1}  \end{bmatrix}
\right\}^{-1}
\begin{bmatrix} 
\boldsymbol{Q}'_1\ast\boldsymbol{V}'_1  \\ 
\boldsymbol{v}'_{T+1}\ast\boldsymbol{q}'_{T+1}  \end{bmatrix}
\begin{bmatrix} 
\left(\boldsymbol{q}_1\boldsymbol{q}'_1+\boldsymbol{Q}_1\boldsymbol{Q}'_1 +\boldsymbol{q}_{T+1}\boldsymbol{q}'_{T+1}\right) \ast\boldsymbol{I}_N\end{bmatrix}^{-1}
.\end{split}
\end{equation}
\end{footnotesize}
Then, pre-multiplying  $(\boldsymbol{q}_1\ast\boldsymbol{v}_1)$ by $\boldsymbol{\Lambda}_{12}$ yields the estimator $\begin{bmatrix}
\underset{(T-1,1)}{\hat{\boldsymbol{\delta}}}
\\\underset{(1,1)}{\hat{\delta}_{T+1}}
\end{bmatrix}$. The reciprocal of the (non-null) elements of this estimator provides the values of the updated multilateral version of the MPL index. 
\end{proof}

\subsection{Proof of Theorem~\ref{th:5}}\label{app:proofs8} 
\begin{proof}
When the values, $\boldsymbol{v}_{T+1}$, and the quantities, $\boldsymbol{q}_{T+1}$, of $N$ commodities of a reference basket become available at time $T+1$, the updating of the multi-period version of the MPL index must not change its past values with meaningful computational advantages. In order to get the required updating formula, let us rewrite Eq.~\eqref{eq:1} as follows
\begin{equation}\label{eq:1proof8}
\left[\underset{(N, T)}{\boldsymbol{V}_{}}\, \, \underset{(N, 1)}{\boldsymbol{v}_{T+1}}\right]
\underset{(T,T)}{\boldsymbol{D}^*_{{{\boldsymbol{\delta}}}}} 
=  \underset{(N, N)}{\boldsymbol{D}_{\tilde{\boldsymbol{p}}}}\, \, \left[\underset{(N, T)}{\boldsymbol{Q}}\,\,\underset{(N,1)}{\boldsymbol{q}_{T+1}}\right] + \left[\underset{(N, T)}{\boldsymbol{E}_{}}\,\,\underset{(N,1)}{\varepsilon_{T+1}}
\right]
\end{equation}
where ${\boldsymbol{D}^*_{{{\boldsymbol{\delta}}}}}$ is specified as follows
\begin{equation}
\underset{(T+1,T+1)}{\boldsymbol{D}^{*}_{{{\boldsymbol{\delta}}}}}=
\begin{bmatrix}
\underset{(T,T)}{\hat{\boldsymbol{D}}_{{{\boldsymbol{\delta}}}}}& \underset{(T,1)}{\boldsymbol{0}}
\\
\underset{(1,T)}{\boldsymbol{0}'}&\underset{(1,1)}{{\delta}_{T+1}}
\end{bmatrix}.
\end{equation}
Here ${\hat{\boldsymbol{D}}_{{{\boldsymbol{\delta}}}}}$ denotes the estimate of ${\boldsymbol{D}}_{{{\boldsymbol{\delta}}}}$, defined as in Eq.~\eqref{eq:1b}, that is 
\begin{equation}
\underset{(T,T)}{\hat{\boldsymbol{D}}_{{{\boldsymbol{\delta}}}}}=
\begin{bmatrix}
\underset{(1,1)}{1}&\underset{(1,T-1)}{\boldsymbol{0}'}\\
\underset{(T-1,1)}{\boldsymbol{0}}&
\begin{bmatrix}
\hat{{\delta}}_2&0&\cdots&0\\
0&\hat{{\delta}}_3&\cdots&0\\
0&0&\cdots&0\\
0&0&\cdots&\hat{{\delta}}_T\\
\end{bmatrix}
\end{bmatrix}
=
\begin{bmatrix}
1&\boldsymbol{0}'\\
0&\hat{\tilde{\boldsymbol{D}_{\boldsymbol{\delta}}}}
\end{bmatrix}
\end{equation}
where the entries $\hat{{\delta}}_2$, $\hat{{\delta}}_3, \dots, {\hat{\delta}}_T$ are the elements of the vector $\hat{{\boldsymbol{\delta}}}$ given in Eq.~\eqref{eq:delta}. 
System in Eq.~\eqref{eq:1proof8} can be also written as
\begin{equation}
\label{eq:2proof8}
\begin{cases}
& \boldsymbol{v}_1=\boldsymbol{D}_{\tilde{\boldsymbol{p}}}\boldsymbol{q}_1+\boldsymbol{\varepsilon}_1
\\& \boldsymbol{V}_1 \hat{\tilde{\boldsymbol{D}}}_{\boldsymbol{\delta}}=\boldsymbol{D}_{\tilde{\boldsymbol{p}}}\boldsymbol{Q}_1+\boldsymbol{E}_1
\\& \boldsymbol{v}_{T+1}{\delta}_{T+1}=\boldsymbol{D}_{\tilde{\boldsymbol{p}}}\boldsymbol{q}_{T+1}+\boldsymbol{\varepsilon}_{T+1}
\end{cases}
\rightarrow\, 
\begin{cases} \boldsymbol{V}\,{\hat{\boldsymbol{D}}}_{\boldsymbol{\delta}} = \boldsymbol{D}_{\tilde{\boldsymbol{p}}} \boldsymbol{Q}+\boldsymbol{E}
\\
\boldsymbol{v}_{T+1}{\delta}_{T+1}=\boldsymbol{D}_{\tilde{\boldsymbol{p}}}\boldsymbol{q}_{T+1}+\varepsilon_{T+1}
\end{cases}.
\end{equation}
The application of the $vec$ operator to the first block of equations in Eq.~\eqref{eq:2proof8} yields
\begin{equation}
\label{eq:gtt}
\begin{cases}
\tilde{\boldsymbol{v}}_1=\left(\boldsymbol{Q}'\otimes \boldsymbol{I}_N\right)\boldsymbol{R}_N' \tilde{\boldsymbol{p}}+\boldsymbol{\eta}\\
\boldsymbol{0}=-\boldsymbol{v}_{T+1}{\delta}_{T+1}+\left(\boldsymbol{q}_{T+1}'\otimes \boldsymbol{I}_N\right)\boldsymbol{R}_N' \tilde{\boldsymbol{p}}+\varepsilon_{T+1}
\end{cases}
\end{equation}
where $\tilde{\boldsymbol{v}}_1=vec(\boldsymbol{V}\,{\hat{\boldsymbol{D}}}_{\boldsymbol{\delta}})=\left( \boldsymbol{I}_T \otimes \boldsymbol{V}\right)\boldsymbol{R}_T' \tilde{{\boldsymbol{\delta}}}$ with $\tilde{{\boldsymbol{\delta}}}'=[1,\hat{{\boldsymbol{\delta}}}]'$ \footnote{Note that, differently from the proof of Theorem~\ref{th:4}, the vector $\boldsymbol{\delta}$ does not enter in the updating estimation process as it is considered given.} and $\boldsymbol{\eta}$ is equal to $vec(\boldsymbol{E})$. 
Eq.~\eqref{eq:gtt} can  be written in vector form as

\begin{equation}\label{eq:reg}
\underset{(N(T+1), 1)}{\boldsymbol{y}_u}=\underset{(N(T+1), N+1)}{\boldsymbol{X}_{u}}\underset{(N+1, 1)}{\boldsymbol{\beta}_u} + \underset{(N(T+1), 1)}{\boldsymbol{\mu}_u}
\end{equation} where 

\begin{equation}\label{eq:reg1}
\underset{(N(T+1),1)}{\boldsymbol{y}_u}=\begin{bmatrix} \tilde{\boldsymbol{v}}_1\\\boldsymbol{0}
\end{bmatrix}, \,\,\underset{(N(T+1),N+1)}{\boldsymbol{X}_{u}}=\begin{bmatrix} \boldsymbol{0} &  (\boldsymbol{Q}'\otimes \boldsymbol{I}_N)\boldsymbol{R}_N'\\ -\boldsymbol{v}_{T+1} & (\boldsymbol{q}_{T+1}'\otimes \boldsymbol{I}_N)\boldsymbol{R}_N'
\end{bmatrix}
\end{equation}
and

\[
\underset{(N+1,1)}{\boldsymbol{\beta}_u}=\begin{bmatrix} {\delta}_{T+1}\\\tilde{\boldsymbol{p}}
\end{bmatrix}, \,\,\underset{(N(T+1),1)}{\boldsymbol{\mu}_u}=\begin{bmatrix} \boldsymbol{\eta} \\ \boldsymbol{\varepsilon}_{T+1}
\end{bmatrix}.
\]

\noindent The OLS estimator of the vector $\boldsymbol{\beta}_u$ is given by
\begin{equation}
\hat{\boldsymbol{\beta}_u}=(\boldsymbol{X}_u'\boldsymbol{X}_u)^{-1}\boldsymbol{X}_u'\boldsymbol{y}_u
\end{equation}
where
\begin{equation}
\label{eq:xx}
\begin{split}
\underset{(N+1,N+1)}{\boldsymbol{X}_u'\boldsymbol{X}_u}&=
\begin{bmatrix} 
\boldsymbol{v}_{T+1}'\boldsymbol{v}_{T+1}& -\boldsymbol{v}_{T+1}'\left(\boldsymbol{q}_{T+1}'\otimes \boldsymbol{I}_N\right)\boldsymbol{R}'_N
\\ -\boldsymbol{R}_{N}(\boldsymbol{q}_{T+1}\otimes (\boldsymbol{I}_N))\boldsymbol{v}_{T+1} & \boldsymbol{R}_{N} ((\boldsymbol{Q}\boldsymbol{Q}'+\boldsymbol{q}_{T+1}\boldsymbol{q}_{T+1}')\otimes \boldsymbol{I}_N)\boldsymbol{R}_N'
\end{bmatrix}
\\
&=\begin{bmatrix}\boldsymbol{v}_{T+1}'\boldsymbol{v}_{T+1} & -\boldsymbol{v}_{T+1}'\ast \boldsymbol{q}_{T+1}'\\
-\boldsymbol{q}_{T+1}\ast\boldsymbol{v}_{T+1} &  ((\boldsymbol{Q}\boldsymbol{Q}'+\boldsymbol{q}_{T+1}\boldsymbol{q}_{T+1}')\ast \boldsymbol{I}_N)
\end{bmatrix}
\end{split}
\end{equation}
and
\begin{equation}
\label{eq:xy}
\underset{(N+1,1)}{\boldsymbol{X}_u'\boldsymbol{y}_u}=
\begin{bmatrix}
\underset{(1,1)}{0}\\\underset{(N,1)}{\boldsymbol{R}_N(\boldsymbol{Q}\otimes \boldsymbol{I}_N)\tilde{\boldsymbol{v}}_1}
\end{bmatrix}
= 
\begin{bmatrix}
\underset{(1,1)}{0}\\
\underset{(N,1)}{\left(\boldsymbol{Q}\ast \boldsymbol{V}\right)\tilde{{\boldsymbol{\delta}}}}
\end{bmatrix}.
\end{equation}
Now, upon nothing that 
\begin{equation}
\\\underset{(1,1)}{\hat{\delta}_{T+1}}=
\begin{bmatrix}
1&\underset{(1,N)}{\boldsymbol{0}'}
\end{bmatrix} \underset{(N+1,1)}{\boldsymbol{\hat{\beta}_u}} = \boldsymbol{\Lambda_{12}}{\left(\boldsymbol{Q}\ast \boldsymbol{V}\right)\tilde{{\boldsymbol{\delta}}}},
\end{equation}
where $\boldsymbol{\Lambda}_{12}$ is the upper off diagonal block of the inverse matrix $(\boldsymbol{X}_u'\boldsymbol{X}_u)^{-1}$
\[\underset{(N+1,N+1)}{(\boldsymbol{X}_u'\boldsymbol{X}_u)^{-1}}
=\boldsymbol{\Lambda}= \begin{bmatrix}\underset{(1,1)}{\boldsymbol{\Lambda}_{11}}&\underset{(1,N)}{\boldsymbol{\Lambda}_{12}}\\
\underset{(N,1)}{\boldsymbol{\Lambda}_{21}} & \underset{(N,N)}{\boldsymbol{\Lambda}_{22}}
\end{bmatrix},
\]
partitioned inversion leads to 
\begin{equation}
\begin{split}
\boldsymbol{\Lambda}_{12}=&\left\{
\boldsymbol{v}_{T+1}'\boldsymbol{v}_{T+1}-\left(\boldsymbol{q}_{T+1}'\ast\boldsymbol{v}_{T+1}'\right)\left[
(\boldsymbol{Q}\boldsymbol{Q}'+\boldsymbol{q}_{T+1}\boldsymbol{q}_{T+1}')\ast\boldsymbol{I}_{N}
\right]^{-1}\left(\boldsymbol{q}_{T+1}\ast\boldsymbol{v}_{T+1}\right)
\right\}^{-1} \\&\left(\boldsymbol{q}_{T+1}'\ast\boldsymbol{v}_{T+1}'\right)\left[
(\boldsymbol{Q}\boldsymbol{Q}'+\boldsymbol{q}_{T+1}\boldsymbol{q}_{T+1}')\ast\boldsymbol{I}_{N}
\right]^{-1}.
\end{split}
\end{equation}
\noindent Then, pre-multiplying ${\left(\boldsymbol{Q}\ast \boldsymbol{V}\right)\tilde{{\boldsymbol{\delta}}}}$ by $\boldsymbol{\Lambda}_{12}$ yields the estimator of $\hat{\delta}_{T+1}$ given in Theorem~\ref{th:5}. The reciprocal of this estimator provides the updated value of the multi-period version of the MPL index. 
\end{proof}

\end{document}